%% file: Main.tex
\documentclass[onecolumn,draftclsnofoot,nofonttune,10pt]{IEEEtran}

\usepackage{amsmath}
\usepackage{amssymb}
\usepackage{amsfonts}
\usepackage{amsthm}
\usepackage{mathtools}
\mathtoolsset{%
}

\usepackage[utf8]{inputenc}
\usepackage[T1]{fontenc}

\usepackage{txfonts}
\let\mathbb=\varmathbb

\usepackage{dsfont}

\usepackage[%
cal=cm,
]
{mathalfa}

\usepackage[dvipsnames,svgnames]{xcolor}
\colorlet{MyBlue}{DodgerBlue!75!Black}
\colorlet{MyGreen}{DarkGreen!85!Black}

\usepackage{subfigure}
\usepackage{tikz}
\usetikzlibrary{calc}
\usetikzlibrary{arrows,backgrounds,calc,decorations.pathmorphing,fit,shapes.arrows,positioning,through}

\usepackage{acronym}
\usepackage{balance}
\usepackage{latexsym}
\usepackage{paralist}
\usepackage{wasysym}
\usepackage{xspace}

\usepackage[numbers,sort&compress]{natbib}

\newcommand{\R}{\mathbb{R}}

\DeclareMathOperator*{\argmin}{arg\,min}

\DeclareMathOperator*{\union}{\bigcup}

\DeclareMathOperator{\bigoh}{\mathcal{O}}

\DeclareMathOperator{\ex}{\mathbb{E}}

\DeclareMathOperator{\intr}{int}

\DeclareMathOperator{\prob}{\mathbb{P}}

\DeclareMathOperator{\supp}{supp}

\newcommand{\dd}{\:d}

\newcommand{\eps}{\varepsilon}

\newcommand{\from}{\colon}
\newcommand{\gen}{\mathcal{L}}

\newcommand{\pd}{\partial}

\newcommand{\wilde}{\widetilde}

\DeclarePairedDelimiter{\braces}{\{}{\}}
\DeclarePairedDelimiter{\bracks}{[}{]}
\DeclarePairedDelimiter{\parens}{(}{)}

\DeclarePairedDelimiter{\abs}{\lvert}{\rvert}
\DeclarePairedDelimiter{\norm}{\lVert}{\rVert}

\DeclarePairedDelimiterX{\braket}[2]{\langle}{\rangle}{#1{}\delimsize\vert{}#2}

\DeclarePairedDelimiterX{\inner}[2]{\langle}{\rangle}{#1,#2}
\DeclarePairedDelimiterX{\setdef}[2]{\{}{\}}{#1:#2}

\DeclarePairedDelimiterXPP{\probof}[1]{\prob}{(}{)}{}{%

#1}

\DeclarePairedDelimiterXPP{\exof}[1]{\ex}{[}{]}{}{%

#1}

\DeclarePairedDelimiterXPP{\expof}[1]{\exp}{(}{)}{}{#1}

\newcommand{\txs}{\textstyle}

\newcommand{\insum}{\sum\nolimits}

\newcommand{\as}{\textup(a.s.\textup)\xspace}

\newtheorem{theorem}{Theorem}
\newtheorem{corollary}{Corollary}
\newtheorem*{corollary*}{Corollary}
\newtheorem{lemma}{Lemma}
\newtheorem{proposition}{Proposition}

\newtheorem{assumption}{Assumption}

\theoremstyle{definition}
\newtheorem{definition}{Definition}
\newtheorem*{definition*}{Definition}

\theoremstyle{remark}
\newtheorem{remark}{Remark}
\newtheorem*{remark*}{Remark}

\newcommand{\feas}{\mathcal{X}}

\newcommand{\dnorm}[1]{\norm{#1}_{\ast}}

\newcommand{\dual}{\mathcal{Y}}

\newcommand{\tcone}{\mathrm{TC}}

\newcommand{\fench}{F}

\newcommand{\gibbs}{G}

\newcommand{\player}{k}

\newcommand{\play}{\player} 		

\newcommand{\nActions}{A}
\newcommand{\actions}{\mathcal{\nActions}}
\newcommand{\act}{\actions}
\newcommand{\strat}{\mathcal{X}}

\newcommand{\cost}{c}
\newcommand{\obj}{C}
\newcommand{\payv}{v}
\newcommand{\pot}{U}
\newcommand{\score}{y}

\newcommand{\eq}{x^{\ast}}
\newcommand{\weq}{\load^{\ast}}

\newcommand{\eqset}{\strat^{\ast}}
\newcommand{\modcost}{\pi}

\newcommand{\km}{\mathrm{km}}

\newcommand{\ghz}{\mathrm{GHz}}
\newcommand{\gbs}{\mathrm{Gb}/\mathrm{s}}
\newcommand{\tbs}{\mathrm{Tb}/\mathrm{s}}
\newcommand{\W}{\mathrm{W}}
\newcommand{\kW}{\mathrm{kW}}

\newcommand{\nodes}{\mathcal{V}}
\newcommand{\edges}{\mathcal{E}}
\newcommand{\source}{s}
\newcommand{\nSources}{S}
\newcommand{\nSinks}{D}
\newcommand{\sources}{\mathcal{\nSources}}
\newcommand{\sink}{d}
\newcommand{\sinks}{\mathcal{D}}

\newcommand{\link}{\mathrm{link}}
\newcommand{\DC}{\mathrm{DC}}
\newcommand{\maxload}{L}

\newcommand{\rate}{\varrho}
\newcommand{\load}{w}
\newcommand{\edge}{e}
\newcommand{\edgealt}{\edge'}

\newcommand{\path}{\alpha}
\newcommand{\pathalt}{\beta}
\newcommand{\nPaths}{A}
\newcommand{\paths}{\mathcal{\nPaths}}

\newcommand{\pmat}{\Pi}
\newcommand{\noisedev}{\sigma}
\newcommand{\noisevar}{\noisedev^{2}}
\newcommand{\noisebound}{\noisedev_{\!\ast}}
\newcommand{\covmat}{\Sigma}
\newcommand{\dkl}{D_{\mathrm{KL}}}
\newcommand{\temp}{\eta}

\newcommand{\acdef}[1]{\emph{\acl{#1}} \textup(\acs{#1}\textup)\acused{#1}}
\newcommand{\acdefp}[1]{\emph{\aclp{#1}} \textup(\acsp{#1}\textup)\acused{#1}}

\begin{document}


\title{Boltzmann Meets Nash:
Energy-Efficient Routing in Optical Networks under Uncertainty}

\author{
Panayotis~Mertikopoulos%
,~\IEEEmembership{Member,~IEEE},
Aris~L.~Moustakas%
,~\IEEEmembership{Senior~Member,~IEEE},
and
Anna~Tzanakaki%
,~\IEEEmembership{Senior~Member,~IEEE}
\thanks{%
P.~Mertikopoulos is with the French National Center for Scientific Research (CNRS) and the Laboratoire d'Informatique de Grenoble (LIG), F-38000, Grenoble, France.
A.~L.~Moustakas and A.~Tzanakaki are with the Physics Department of the National \& Kapodistrian University of Athens, Athens, Greece.
Part of this work appeared in the conference paper \cite{MM11} and the report \cite{MM12}.
}
} 

\maketitle

\newacro{ICT}{information and communications technology}
\newacro{IoT}{Internet of Things}
\newacro{VoD}{video on demand}
\newacro{AWGN}{additive white Gaussian noise}
\newacro{WDM}{wavelength-division multiplexing}
\newacro{DC}{data center}
\newacro{NE}{Nash equilibrium}
\newacroplural{NE}[NE]{Nash equilibria}
\newacro{WE}{Wardrop equilibrium}
\newacroplural{WE}[WE]{Wardrop equilibria}
\newacro{KL}{Kullback\textendash Leibler}
\newacro{KKT}{Karush\textendash Kuhn\textendash Tucker}
\newacro{PoA}{price of anarchy}
\newacro{MFDL}{minimum first derivative length}
\newacro{SDE}{stochastic differential equation}
\newacro{ISP}{Internet service provider}
\acused{ISP}
\newacro{od}[O/D]{origin-destination}
\newacro{OSPF}{Open Shortest Path First}
\newacro{RTT}{round-trip time}
\newacro{RHS}{right hand side}
\acused{RHS}
\newacro{LHS}{left hand side}
\acused{LHS}
\newacro{CDF}{cumulative distribution function}
\newacro{ESS}{evolutionarily stable strategy}
\newacro{NSS}{neutrally stable strategy}
\newacro{wlog}[w.l.o.g.]{without loss of generality}
\acused{wlog}

\begin{abstract}
\input{Abstract}
\end{abstract}

\begin{IEEEkeywords}
\vspace{-2ex}
Boltzmann routing;
energy efficiency;
fluctuations;
optical networks;
uncertainty;
Nash equilibrium.
\end{IEEEkeywords}


\acresetall

\section{Introduction}
\label{sec:intro}
\input{Introduction}

\section{System Model and Problem Formulation}
\label{sec:model}
\input{Model}

\section{Boltzmann Routing}
\label{sec:learning}
\input{Learning}

\section{Convergence Analysis}
\label{sec:analysis}
\input{Analysis}

\section{Numerical Results}
\label{sec:numerics}
\input{Numerics}

\section{Conclusions and Perspectives}
\label{sec:conclusions}
\input{Conclusions}

\appendices
\numberwithin{equation}{section}
\numberwithin{theorem}{section}
\numberwithin{proposition}{section}
\numberwithin{lemma}{section}

\section{Preliminary Results}
\label{app:basics}
\input{App-Basics}

\section{Deterministic Analysis}
\label{app:deterministic}
\input{App-Deterministic}

\section{Convergence of Long-Term Averages}
\label{app:averages}
\input{App-Averages}

\section{Convergence to Non-Mixing Optimum States}
\label{app:strict}
\input{App-Strict}
\section{Concentration near Interior Optimum States}
\label{app:interior}
\input{App-Interior}

\bibliographystyle{IEEEtran}
\bibliography{IEEEabrv,Bibliography_alm,Bibliography}

\end{document}

%% file: Abstract.tex
%
%
\thispagestyle{empty}
\vspace{-2ex}
Motivated by the massive deployment of power-hungry \aclp{DC} for IT service provisioning, we examine the problem of routing in optical networks with the aim of minimizing traffic-driven power consumption.
To tackle this issue, routing must take into account energy efficiency as well as capacity considerations;
moreover, in rapidly-varying network environments, this must be accomplished in a real-time, distributed manner that remains robust in the presence of random disturbances and noise.
In view of this, we derive a pricing scheme whose \aclp{NE} coincide with the network's socially optimum states, and we propose a distributed learning method based on the Boltzmann distribution of statistical mechanics.
Using tools from stochastic calculus, we show that the resulting Boltzmann routing scheme exhibits remarkable convergence properties under uncertainty:
specifically, the long-term average of the network's power consumption converges within $\eps$ of its minimum value in time which is at most $\tilde\bigoh(1/\eps^{2})$, irrespective of the fluctuations' magnitude;
additionally, if the network admits a strict, non-mixing optimum state, the algorithm converges to it \textendash\ again, no matter the noise level.
Our analysis is supplemented by extensive numerical simulations which show that Boltzmann routing can lead to a significant decrease in power consumption over basic, shortest-path routing schemes in realistic network conditions.

%% file: Introduction.tex

\IEEEPARstart{I}{n} recent years, the Internet has become the cornerstone of global telecommunications, to the point that it now plays an integral part in sustaining the growth of the world economy.
Specifically, as more and more users are using the Internet for backing up data, \ac{VoD}, device synchronization and other IT services, data-hungry applications are experiencing a wildfire growth that affects almost all aspects of modern-day, always-connected societies.
As a result, optical networks must grow fast enough so as to avoid experiencing a ``capacity crunch'' that will cripple their ability to handle the requested traffic in a reliable and efficient manner \cite{Tkach2010_ScalingCrunch}.

At the same time however, energy consumption in data networks must also be kept in check:
in total, power consumption from \ac{ICT} applications has been estimated to $4\%$ of the global energy consumption \cite{PVD+08}, with a steeply rising trend as we progress toward the \ac{IoT} era.
In particular, just in the United States, Internet equipment consumes a staggering $8\%$ of the total energy with a predicted growth of $50\%$ within a decade \cite{Plepsys2002_TheGreySideOfICT}.
Therefore, in addition to the always-present requirements of minimizing latencies and maximizing throughput, there is a key developing need to operate optical networks in an energy-aware way that minimizes power expenditure to the bare minimum.
This last goal becomes especially prominent when taking into account the power needs of large switching and \acp{DC}
which are rapidly becoming a dominant term in the power consumption equation due the huge amounts of traffic that they service \cite{Vukovic2005_Network,Fettweis2008_ICT}.

In this paper, motivated by the massive deployment of power-hungry \aclp{DC} described above, we examine the problem of routing traffic with the aim of minimizing the network's traffic-driven power consumption.
To achieve this, centralized optimization is not an option, because it is neither scalable nor flexible.
Instead, one needs distributed load management as close as possible to the link level while ensuring that routing choices strike an optimal balance between ``greener'', more efficient \textendash\ but possibly congested \textendash\ network facilities and ``hungrier'', legacy resources that may be under-utilized.
To tackle this issue, we derive a game-theoretic pricing scheme \cite{San07} which reconciles the selfish solution concept of (nonatomic) \acl{NE} with social optimality (i.e. states that minimize global power consumption in the network).
Building on this equivalence, we further propose a distributed learning scheme based on the Botlzmann distribution of statistical mechanics,
which bears strong ties to replicator-based methods \cite{BK03,BEL06,KDB14,FV04} and the class of exponential learning schemes used in finite (atomic) games \cite{MM10,MBML12,DMMP15}.

Exploiting these links, we are able to show that the proposed routing scheme converges rapidly to a socially efficient state \textendash\ in practice, within a few iterations.
However, a key limiting factor in these considerations is that data traffic may exhibit significant spatiotemporal fluctuations because different applications have highly variable characteristic time scales.
In addition, obtaining perfect information on the energy consumption of each switch/\acl{DC} is a very challenging task in large networks
due to real-time propagation delays,
hierarchical topology aggregation and update sparsity (so as to keep the signaling overhead in check)
\cite{Zhou2002_Study, Masip2003_Routing, Yannuzzi2009_Performance}.
This uncertainty is exacerbated further in highly dynamic networks where state information becomes obsolete very fast;
accordingly, it is not clear whether a routing scheme designed for static, deterministic networks can maintain its convergence properties under uncertainty and noise.

To account for this crucial challenge, we focus on the performance of the proposed Boltzmann routing method in the presence of unpredictable stochastic disturbances.
Using tools from stochastic analysis and Itô calculus \cite{Oks07,Kuo06}, we show that Boltzmann routing is exceptionally robust in this regard:
first, by finetuning the choice of the method's so-called \emph{inverse temperature}, we show that the network's long-term average consumption converges within $\eps$ of its minimum value in at most $\tilde\bigoh(1/\eps^{2})$ time, irrespective of the fluctuations' magnitude.
Moreover, if the network admits a strict, non-mixing optimum state, the method converges to it almost surely \as, no matter the degree of volatility.
Finally, if the network instead admits an interior, fully-mixing optimum state, the proposed scheme remains arbitrarily close to said state with probability arbitrarily close to $1$, provided that the scheme's inverse temperature parameter is taken sufficiently small.

The performance of Boltzmann routing is validated by means of extensive numerical simulations modeling an optical data network deployed over the largest metropolitan centers of the continental US.
Our results show that Boltzmann routing represents a scalable and flexible method that reduces traffic-driven energy consumption by a factor of $40\%$ over simple shortest-path/closest-destination routing modes, even under high degrees of volatility and uncertainty.



\subsection*{Paper outline}

In Section \ref{sec:model}, we present our network energy consumption model and a Pigouvian pricing scheme based on the theory of nonatomic congestion games.
Subsequently, in Section \ref{sec:learning}, we introduce the proposed Boltzmann routing method, first in a deterministic and then in a fully stochastic environment.
Our main theoretical results (outlined above) are presented in Section \ref{sec:analysis}, in both a deterministic and a stochastic context.
Finally, Section \ref{sec:numerics} contains our numerical simulations in practical optical networks with diverse link, switching, and \acl{DC} specifications.
To streamline our presentation, technical proofs have been relegated to a series of appendices at the end.



%% file: Model.tex

In this section, we describe our model for power consumption in optical networks where traffic demands are routed towards a set of service provisioning infrastructures \textendash\ the network's \acdefp{DC}.
After a few preliminaries in Section \ref{sec:prelims}, we discuss the specifics of power consumption in Section \ref{sec:consumption} and we present our game-theoretic model and pricing scheme in Section \ref{sec:equilibrium}.

\subsection{Preliminaries}
\label{sec:prelims}

Following the standard traffic model of \cite{BG92,Ros73,Rou05,San10,NRTV07}, we begin with a finite undirected graph with vertex set $\nodes$ and edge set $\edges$, describing respectively the network's \emph{nodes} and \emph{links}.
We assume that traffic in the network is generated in continuous streams by a set of \emph{source nodes} $\sources\subseteq\nodes$, each node $\source\in\sources$ generating traffic demands at rate $\rate_{\source} > 0$ towards a set of \emph{destination nodes} $\sinks_{\source}\subseteq\nodes$.
In our setting, the traffic generated at each node represents the aggregation of the demands of a very large number of infinitesimal, nonatomic users \textendash\ for instance, a city's smartphone users accessing a cloud-based service.
By the same token, the destination set $\sinks_{\source}$ represents the network's \aclp{DC} (or a subset thereof), so each traffic element generated at source $\source$ is to be delivered in an \emph{anycast} fashion to any one of the destination nodes in $\sinks_{\source}$ (for a schematic representation, see Fig.~\ref{fig:network}).%
\footnote{This generic setup can easily account for a variety of different delivery semantics, such as unicast or multicast routing.
For instance, in the unicast case, each infinitesimal traffic element generated at $\source$ would simply be addressed to a \emph{specific} destination node in $\sinks_{\source}$ instead of any $\sink\in\sinks_{\source}$ \cite{Rou05,San10}.
Essentially, switching to a different delivery scheme is just a matter of notation;
given our motivation, we focus on anycast routing throughout, but our results also apply to these different routing schemes.}

Based on the routing choices of each individual traffic element, this traffic is split over a small, relevant set $\paths_{\source}$ of \emph{paths} (or \emph{routes}) that have been chosen offline so as to join source $\source\in\sources$ to the associated destination set $\sinks_{\source}$ (for instance, all paths with the lowest hop-count).
As such, if $x_{\path}$ denotes the amount of traffic routed from source $\source\in\sources$ via the path $\path\in\paths_{\source}$, the corresponding \emph{flow profile} $x_{\source} = (x_{\path})_{\path\in\paths_{\source}}$ will be an element of the scaled simplex
\begin{equation}
\txs
\feas_{\source}
	= \setdef[\big]{x_{\source}\in\R^{\paths_{\source}}}{\text{$\sum_{\path\in\paths_{\source}} x_{\path} = \rate_{\source}$ and $x_{\path} \geq 0$ for all $\path\in\paths_{\source}$}}.
\end{equation}
Coalescing all these individual flow distributions into a single state variable, the aggregate profile $x = (x_{\source})_{\source\in\sources} \in \prod_{\source} \R^{\paths_{\source}}$ will be referred to as the network's \emph{state} and the set $\feas \equiv \prod_{\source} \feas_{\source}$ of all such states will be called the \emph{state space} of the network.
For concision, it will often be convenient to write $x = (x_{\path})_{\path\in\paths}$ where $\paths = \union_{\source} \paths_{\source}$ denotes the set of all utilized paths in the network;
since a path $\path\in\paths$ uniquely determines its source and destination nodes, there is no risk of confusion.

\begin{figure}[t]
\centering
\includegraphics[width=.48\textwidth]{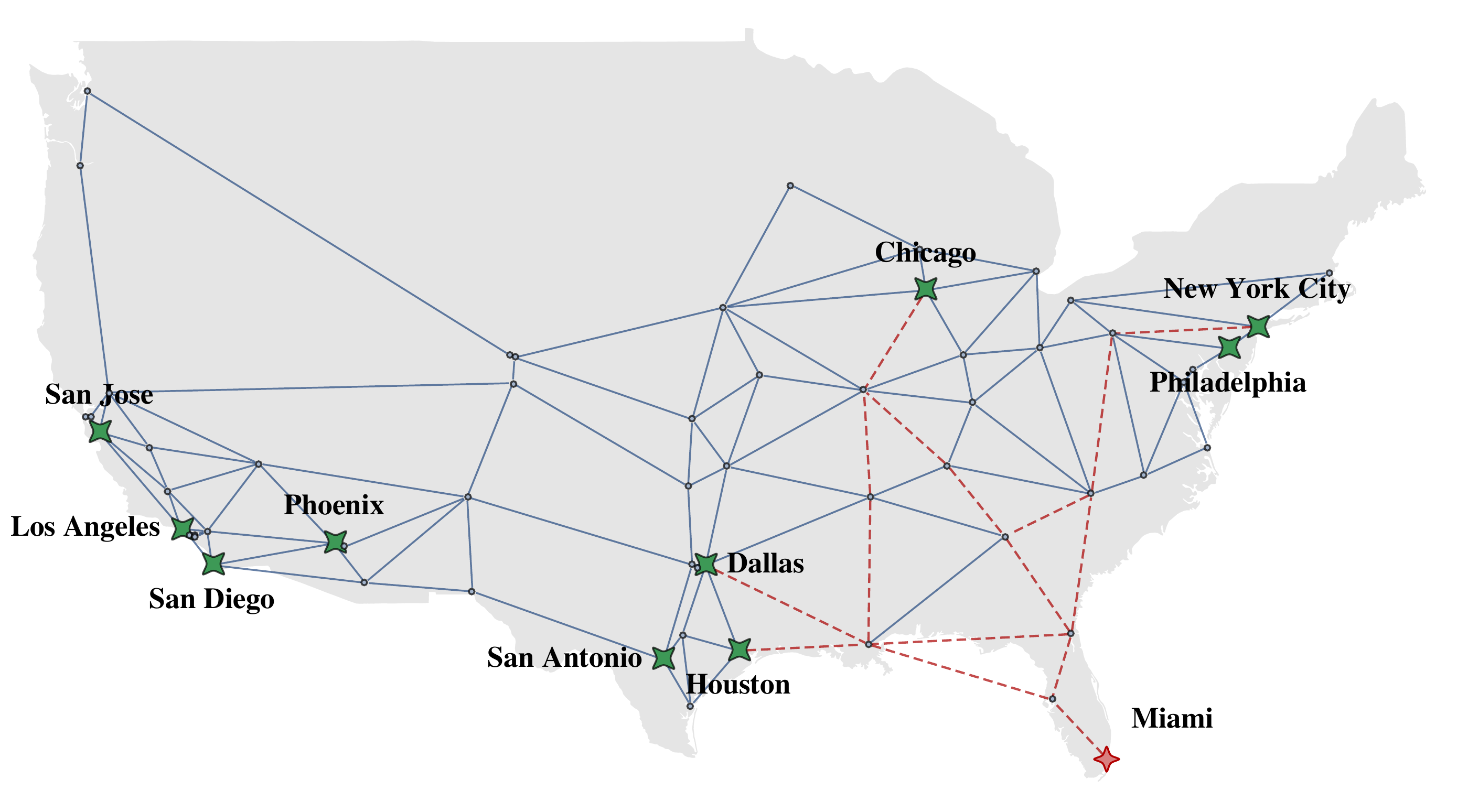}
\vspace{-3ex}
\caption{A typical network topology with source nodes and \aclp{DC} deployed over the continental US.
For illustration purposes, we have highlighted a source node $\source\in\sources$ (red diamond) and its routing choices $\path\in\paths_{\source}$ (dashed lines).}
\label{fig:network}
\vspace{-1ex}
\end{figure}

Now, if the network is at state $x\in\feas$, the \emph{load} $\load_{\edge}$ on edge $\edge\in\edges$ is defined as the sum of all flows through $\edge$, namely
\begin{equation}
\label{eq:load}
\load_{\edge}
	= \sum_{\substack{\path\in\paths\\ \path\ni\edge}} x_{\path}
	= \sum_{\path\in\paths} \pmat_{\edge\path} x_{\path},
\end{equation}
where
$\pmat = (\pmat_{\edge\path})_{\edge\in\edges,\path\in\paths}$ denotes the \emph{edge-path incidence matrix} of the network, viz.
\begin{equation}
\label{eq:incidence}
\pmat_{\edge\path}=
	\begin{cases}
		1
		&\quad
		\text{if $\edge\in\path$},
		\\
		0
		&\quad
		\text{otherwise.}
	\end{cases}
\end{equation}
With this in mind, we associate to each edge $\edge\in\edges$ a \emph{consumption} (or \emph{cost}) \emph{function} $\cost_{\edge}\from\R_{+}\to\R$ reflecting the negative externalities induced on edge $\edge$ due to its traffic load $\load_{\edge}$.
In realistic network scenarios, the externalities incurred on a given edge $\edge\in\edges$ increase (or, at best, do not decrease) with the underlying load $\load_{\edge}$;
furthermore, these externalities accrue at an increased rate as the load increases.
On that account, our only assumption will be:

\begin{assumption}
\label{asm:cost}
The network's consumption functions $\cost_{\edge}\from\R_{+}\to\R$ are convex and nondecreasing.
\end{assumption}

\begin{remark}
The implications of Assumption \ref{asm:cost} in the context of power consumption in optical networks are detailed in the following section;
we only state here that this hypothesis is sufficiently broad so as to accommodate a wide variety of congestion-limited frameworks, ranging from latency minimization \cite{BG92,Rou05,NRTV07} to urban traffic management \cite{Ros73,San10}, etc.%
\footnote{For instance, assuming M/M/1 service queues in a latency minimization framework, a standard choice for $\cost_{\edge}$ is $\cost_{\edge}(\load_{\edge}) = \load_{\edge}/(\mu_{\edge} - \load_{\edge})$, where $\mu_{\edge}$ denotes the capacity of edge $\edge\in\edges$ \cite{BG92,Rou05}.}
For this reason, all our analysis will be carried out under the general umbrella of Assumption \ref{asm:cost}.
\end{remark}

With all this at hand, we obtain the global \emph{consumption minimization problem}
\begin{equation}
\label{eq:CM}
\tag{CM}
\begin{aligned}
\text{minimize}
	&\quad
	\obj(x) = \sum_{\edge\in\edges} \cost_{\edge}(\load_{\edge})
	\\
\text{subject to}
	&\quad
	\load = \pmat x,\;
	x\in\feas.
\end{aligned}
\end{equation}
In words, the network states $\eq$ that solve \eqref{eq:CM} minimize the aggregate consumption $\obj(x)$ over the entire network;
as such, the solutions of \eqref{eq:CM} will be called the network's \emph{socially optimum states}.
As we show below, the set of such states has a particularly simple structure under Assumption \ref{asm:cost}:

\begin{proposition}
\label{prop:structure}
The solution set $\eqset$ of \eqref{eq:CM} is nonempty, convex and compact.
Moreover, if the network's consumption functions are strictly increasing and convex, every solution of \eqref{eq:CM} induces the same load profile $\weq = (\weq_{\edge})_{\edge\in\edges}$;
finally, if the utilized paths $\path\in\paths$ are also linearly independent \textup(in the sense that the network's path-edge incidence matrix $\pmat$ is invertible on $\feas$\textup), the global minimization problem \eqref{eq:CM} admits a unique solution $\eq \in \feas$.
\end{proposition}

\begin{IEEEproof}
The first part of our claim is a classic result that follows immediately from the compactness of $\strat$ and the convexity of each $\cost_{\edge}$, $\edge\in\edges$ \textendash\ see e.g. \cite{BMW56,Rou05}.
For our second assertion, simply note that if the network's cost functions are strictly convex and increasing, $\obj$ is strictly convex as a function of $\load$, so it admits a unique solution $\weq$.
Hence, if $\pmat$ is invertible on $\feas$, the equation $\weq = \pmat \eq$ admits a unique solution $\eq\in\feas$, as claimed.
\end{IEEEproof}

The above result shows that the network's set of socially optimum states has some fairly desirable attributes;
however, it does not provide a way for the network's users to compute (or otherwise converge to) such a state in a distributed way.
This question underlies much of this paper's motivation, so we address it in detail in Sections \ref{sec:learning} and \ref{sec:analysis}.

\subsection{Energy consumption in optical data networks}
\label{sec:consumption}

As we mentioned in the introduction, traffic-driven power consumption in optical data networks occurs at two basic levels:
\begin{inparaenum}
[\upshape(\itshape i\upshape)]
\item
at the \emph{link level} (including all wave\-length-related and transmission line elements);
and
\item
at the \emph{\acl{DC} level} (including all costs to service, process and/or store demands at the network's \aclp{DC}, cooling, etc.)
\end{inparaenum}
\cite{FWB07,BGT+13}.
We will thus consider two types of power consumption functions:

\begin{subequations}
\label{eq:cost-model}
\begin{enumerate}
\setlength{\parindent}{1em}

\item
\emph{Link-based consumption.}
At the transmission line, power consumption typically comprises two parts:
\begin{inparaenum}
[\itshape a\upshape)]
\item
a \emph{traffic-independent} component (due to amplification and other factors);
and
\item
a \emph{traffic-dependent} component, proportional to the number of active wavelengths per link (including the cost of switching, electrical-to-optical and optical-to-electrical conversion, etc.).
\end{inparaenum}
Assuming that wavelength granularity is sufficiently fine, the number of wavelengths needed to carry a traffic load $\load$ is itself proportional to $\load$, leading to the link-based model
\cite{Anastasopoulos2011_Evolutionary,BGT+13}
\begin{equation}
\label{eq:cost-link}
\cost_{\link}(\load)
	= A_{\link} + B_{\link} \, \load,
\end{equation}
where the values of $A_{\link}$ and $B_{\link}$ depend on the power specifications of each link.

\item
\emph{Destination-based consumption.}
At the \acl{DC} level, power consumption again comprises two parts:
\begin{inparaenum}
[\itshape a\upshape)]
\item
a traffic-\emph{independent} element (due to cooling, infrastructure maintenance, etc.);
and
\item
a traffic-\emph{dependent} part, proportional to the induced CPU load \textendash\ and hence, proportional to the traffic load at each \acl{DC}.
We thus obtain the consumption model \cite{FWB07}:
\begin{equation}
\label{eq:cost-DC}
\cost_{\DC}(\load)
	= A_{\DC} + B_{\DC} \, \load.
\end{equation}
\end{inparaenum}

Of course, the consumption model \eqref{eq:cost-DC} is node-based, so the edge-based formulation of the previous section does not immediately apply.
However, since \aclp{DC} are \emph{terminal} destination nodes, this can be remedied as follows:
First, adjoin to each \acl{DC} node $\sink$ a virtual node $\sink'$ and a virtual edge $\edge = (\sink,\sink')$;
assume further that all traffic reaching $\sink$ is rerouted to $\sink'$ via $\edge$.
Then, the load at $\sink$ can be represented by the load on $\edge = (\sink,\sink')$ and the cost function of this fictitious edge reflects the energy consumption at the \acl{DC}.
\end{enumerate}
\end{subequations}

In addition to the above, in practical optical networks, each link/\acl{DC} has a finite capacity, leading to the state constraint
\begin{equation}
\label{eq:capacity}
\load_{\edge}
	\leq \maxload_{\edge}
	\quad
	\text{for all $\edge\in\edges$},
\end{equation}
where $\maxload_{\edge}$ represents the \emph{maximum load} supported on edge $\edge$.
We thus obtain the \emph{capacity-constrained consumption minimization problem}
\begin{equation}
\label{eq:CM0}
\tag{\ref{eq:CM}$_{0}$}
\begin{aligned}
\text{minimize}
	&\quad
	\obj(x),
	\\
\text{subject to}
	&\quad
	x\in\feas_{0},
\end{aligned}
\end{equation}
where
\begin{equation}
\label{eq:feasible}
\feas_{0}
	= \setdef{x\in\feas}{\load_{\edge} = (\pmat x)_{\edge} \leq \maxload_{\edge}\; \text{for all $\edge\in\edges$}}
\end{equation}
denotes the network's \emph{capacity-constrained feasible region}.
Clearly, if the traffic generation rates $\rate_{\source}$ are very high, the set $\feas_{0}$ may be empty, in which case \eqref{eq:CM0} does not admit a feasible solution.
It will therefore be useful to introduce the following $\eps$-relaxation of \eqref{eq:CM0}:
\begin{equation}
\label{eq:CM-eps}
\tag{\ref{eq:CM}$_{\eps}$}
\begin{aligned}
\text{minimize}
	&\quad
	\obj_{\eps}(x) = \sum_{\edge\in\edges} \cost_{\edge}^{\eps}(\load_{\edge})
	\\
\text{subject to}
	&\quad
	\load = \pmat x,\;
	x\in\feas,
\end{aligned}
\end{equation}
where
the \emph{$\eps$-adjusted consumption functions} $\cost_{\edge}^{\eps}$ are given by
\begin{equation}
\label{eq:cost-eps}
\cost_{\edge}^{\eps}(\load)
	= \begin{cases}
	\cost_{\edge}(\load)
		&\quad
		\text{if $\load \leq \maxload_{\edge}$},
		\\
	\cost_{\edge}(\maxload_{\edge}) + (\load - \maxload_{\edge})/\eps
		&\quad
		\text{if $\load \geq \maxload_{\edge}$}.
	\end{cases}
\end{equation}

In words, the consumption model \eqref{eq:cost-eps} of \eqref{eq:CM-eps} coincides with that of \eqref{eq:CM0} up to the edge's maximum load $\maxload_{\edge}$ and then increases sharply with slope $1/\eps$.
Since $\cost_{\edge}^{\eps}$ is strictly increasing, any feasible solution of \eqref{eq:CM0} will also be a solution of \eqref{eq:CM-eps}.
On the other hand, given that $\feas$ is compact and nonempty, \eqref{eq:CM-eps} always admits a solution, even when \eqref{eq:CM0} does not;
in this case, the solutions of \eqref{eq:CM-eps} can be considered as ``approximate'' solutions to \eqref{eq:CM0} in the limit $\eps\to0$.
We formalize this in the following proposition:

\begin{proposition}
\label{prop:CM-eps}
Let $\obj_{\eps}^{\ast} = \min\setdef{\obj_{\eps}(x)}{x\in\feas}$ denote the minimum value of the global minimization problem \eqref{eq:CM-eps}, and let $\obj_{0}^{\ast}$ denote the corresponding quantity for the capacity-constrained problem \eqref{eq:CM0}.
Then,
$\lim_{\eps\to 0} \obj_{\eps}^{\ast} = \obj_{0}^{\ast}$ \textup(with the standard convention $\min\varnothing = \infty$\textup).
\end{proposition}

\begin{IEEEproof}
Write $\feas_{\eps} = \feas$ for the feasible set of \eqref{eq:CM-eps}, $\eps>0$.
By construction, the set-valued correspondence $\eps\mapsto\feas_{\eps}$, $\eps\geq0$, is upper hemicontinuous \cite{AB99} and the objective function of \eqref{eq:CM-eps} coincides with that of \eqref{eq:CM} on $\feas_{0}$.
Thus, by a precursor to Berge's maximum theorem \cite[Lemma 16.30]{AB99}, it follows that
the function $\eps\mapsto \obj_{\eps}^{\ast}$, $\eps\geq0$, is lower semicontinuous.
Since $\lim_{\eps\to0} \cost_{\edge}^{\eps}(\load) = \infty$ for all $\load>\maxload_{\edge}$, our claim is immediate.
\end{IEEEproof}

\smallskip

In what follows, we will assume for simplicity that the capacity-constrained feasible region $\feas_{0}$ of \eqref{eq:CM0} is nonempty (after all, there is nothing to optimize if $\feas_{0}$ is empty).
In this case, using the $\eps$-relaxation \eqref{eq:CM-eps} allows us to recover the solutions of \eqref{eq:CM0} in the limit $\eps\to0$.
Given that the $\eps$-adjusted consumption functions $\cost_{\edge}^{\eps}$ satisfy Assumption \ref{asm:cost},
this will allow us to apply the general analysis for \eqref{eq:CM}, even in the presence of capacity constraints of the form \eqref{eq:capacity}.

\subsection{Pigouvian pricing and \acl{NE}}
\label{sec:equilibrium}

In a fully distributed environment, the system designer would like to ensure that the network's nonatomic (infinitesimal) users make socially optimal routing choices while unilaterally minimizing the price they pay for accessing the network and utilizing its resources.
Following Pigou's theory of pricing \cite{Pig20}, our approach to achieve this will be to charge agents for the externalities that they induce at a socially optimum state;
in so doing, one can then ensure that such states constitute an equilibrium of the agents' selfish interactions \cite{San07,NRTV07}.

To make this precise, suppose that a user that sends an infinitesimal amount of traffic $dx$ through edge $\edge\in\edges$ is charged $\modcost_{\edge}(\load_{\edge}) \dd x$ where $\modcost_{\edge}(\load_{\edge})$ is the \emph{price per unit of traffic} on edge $\edge$ at load $\load_{\edge}$.
The total amount charged over path $\path\in\act$ will then be
\begin{equation}
\label{eq:modcost-path}
\modcost_{\path}(x) \dd x
	= \sum_{\edge\in\path} \modcost_{\edge}(\load_{\edge}) \dd x,
\end{equation}
so a (nonatomic) user will be satisfied with his routing choice if the charged price $\modcost_{\alpha}(x)$ is the lowest among all other available paths/destinations \cite{War52}.
Formally:

\begin{definition}
\label{def:Nash}
A state $\eq\in\feas$ is at \emph{\acl{NE}} if,
for all paths $\path\in\paths$ with $\eq_{\path}>0$, we have
\begin{equation}
\label{eq:Nash}
\tag{NE}
\modcost_{\path}(\eq)
	= \min_{\pathalt\sim\path} \modcost_{\pathalt}(\eq),
\end{equation}
where the minimum is taken over all paths $\pathalt\sim\path$ originating at the same source as $\path$.
In  words, \emph{$\eq$ is a \acl{NE} when every nonatomic user chooses the least expensive path.}

We will also say that $\eq$ is an \emph{interior} (or \emph{fully-mixing}) \acl{NE} if $\eq\in\intr(\feas)$, i.e. if it employs all paths $\path\in\paths$;
by contrast, $\eq$ will be called \emph{strict} (or \emph{non-mixing}) if $\argmin_{\path\in\paths_{\source}} \modcost_{\path}(\eq)$ is a singleton for every source node $\source\in\sources$.
\end{definition}


A standard result in the theory of nonatomic congestion games is that \aclp{NE} can be characterized as the solutions of a certain convex program.
Specifically, following \cite{BMW56,Ros73,San10,Rou05,DS69}, consider the \emph{potential function}
\begin{equation}
\label{eq:pot}
\pot(x)
	= \sum_{\edge\in\edges} \pot_{\edge}(\load_{\edge})
\end{equation}
where $\load = \pmat x$ and
\begin{equation}
\label{eq:pot-edge}
\pot_{\edge}(x)
	= \int_{0}^{\load_{\edge}} \modcost_{\edge}(w) \dd w.
\end{equation}
Then, as was shown in \cite{BMW56,DS69}, $\eq$ is a \acl{NE} if and only if it solves the \emph{potential minimization} problem
\begin{equation}
\label{eq:PM}
\tag{PM}
\eq
	\in \argmin_{x\in\feas} \pot(x).
\end{equation}
Hence, by comparing \eqref{eq:PM} and \eqref{eq:CM}, we obtain the following Pigouvian pricing scheme:

\begin{proposition}
\label{prop:Pigou}
If the charged price per unit of traffic on edge $\edge$ is $\modcost_{\edge}(\load_{\edge}) = \pd_{-} \cost_{\edge}(\load_{\edge})$, \aclp{NE} coincide with the solutions of the global consumption minimization problem \eqref{eq:CM}.
\end{proposition}

\begin{remark}
In the above, $\pd_{-} \cost_{\edge}(\load_{\edge})$ denotes the left derivative of $\cost_{\edge}$ at $\load_{\edge}$.
This derivative always exists because $\cost_{\edge}$ is assumed convex \cite{Roc70};
if $\cost_{\edge}$ is smooth, we can simply take $\modcost_{\edge} = \cost_{\edge}'$.
\end{remark}

Proposition \ref{prop:Pigou} describes the pricing scheme that the system designer should use so that selfish, cost-minimizing users end up minimizing global power consumption at equilibrium.
As such, in the rest of this paper, we will assume that prices are set in accordance to Proposition \ref{prop:Pigou}, so that socially optimum states for \eqref{eq:CM} coincide with the network's \aclp{NE} under $\modcost$.


%% file: Learning.tex

In this section, we present a distributed routing scheme to attain a socially efficient state based on a logit-type learning rule inspired from statistical mechanics.
This scheme consists of two basic steps:
First, each source node $\source\in\sources$ keeps a running ``score'' for each path starting at $\source$ by aggregating the associated end-to-end price over time.
Then, every nonatomic user based at $\source$ chooses a path (and hence, a destination node $\sink\in\sinks_{\source}$) with probability inversely proportional to the exponential of this performance score \textendash\ in analogy to the Boltzmann distribution of statistical mechanics.

For generality, throughout this section, we focus on the global minimization problem \eqref{eq:CM} with consumption functions satisfying Assumption \ref{asm:cost};
the specifics of energy efficiency in optical data networks are discussed in Section \ref{sec:numerics}.
To streamline our presentation, we first develop our routing scheme in a deterministic setting with perfect information.
Subsequently, we extend our model to a fully stochastic setting where the network is subject to random perturbations (arising e.g. from fluctuations in background traffic, link quality, estimation errors, etc.).
We then present our convergence analysis and theoretical results in Section \ref{sec:analysis}.

\subsection{Boltzmann routing}
\label{sec:deterministic}

The first step in setting up our routing scheme is to introduce a path's ``score'', interpreted here as a cumulative measure of the path's price over time.
More precisely, if $x(t)$ denotes the state of the system at time $t$, we define the \emph{score} of path $\path\in\paths$ as
\begin{equation}
\label{eq:score}
\score_{\path}(t)
	= \score_{\path}(0) + \int_{0}^{t} \modcost_{\path}(x(s)) \dd s,
\end{equation}
i.e. as the cumulative price of path $\path$ over the time interval $[0,t]$ (the value $\score_{\path}(0)$ is an arbitrary constant that represents an initial assessment of the path's price).

Clearly, high scores indicate commensurately high prices, so such paths should be selected with low probability.
We thus posit that a nonatomic user starting at $\source$ selects path $\path\in\paths_{\source}$
with probability given by the \emph{Boltzmann distribution}
\begin{equation}
\label{eq:Boltzmann}
\tag{B}
P_{\!\path}(t)
	\propto \expof{-\temp(t)\,\score_{\path}(t)},
\end{equation}
where $\temp(t) > 0$ is an \emph{inverse temperature} parameter whose role is explained below.
Thus, writing \eqref{eq:score} in differential form and invoking the law of large numbers to obtain the flow induced by \eqref{eq:Boltzmann}, we obtain the \emph{Boltzmann routing scheme}
\begin{equation}
\label{eq:BR}
\tag{BR}
\begin{aligned}
\dot \score_{\path}
	&= \modcost_{\path}(x),
	\\
x_{\path}
	&= \frac{\rate_{\source} \expof{-\temp\score_{\path}}}{\sum_{\pathalt\in\paths_{\source}} \expof{-\temp\score_{\pathalt}}},
\end{aligned}
\end{equation}
where $\source \equiv \source(\path)$ denotes the starting point of $\path$ and $\rate_{\source}$ is the corresponding traffic generation rate.
These routing dynamics will be at the core of our paper, so a few remarks are in order:

\paragraph{Basic properties}
The Boltzmann scheme \eqref{eq:BR} enjoys the following desirable properties:
\begin{enumerate}[(P1)]
\item
It is \emph{consistent}:
$x_{\path}(t)\geq0$ and $\sum_{\path\in\paths_{\source}} x_{\path}(t)=\rate_{\source}$ for all $t\geq0$, so $x(t)$ is a valid state variable.

\item
It is \emph{reinforcing}:
traffic elements tend to be routed along paths with lower prices.

\item
It is \emph{stateless}:
routing choices do not require knowledge of the network's state.

\item
It is \emph{distributed}:
each node only needs to monitor the price of the paths that start at said node.
\end{enumerate}

We should also note that \eqref{eq:BR} has fairly low requirements in terms of computational complexity.
Specifically, each node $\source\in\sources$ only needs to run a cheap update on $\abs{\paths_{\source}}$ variables;
since the set of paths being utilized by each source node is typically small (for instance, those with the minimum hop count), every node $\source\in\sources$ only needs to process at most $\bigoh(\sinks_{\source}) \leq \bigoh(\nodes)$ variables.
By comparison, the application of distributed flow-deviation methods \cite{BG92} would require each node to keep track of $\bigoh(\edges)$ link variables and then run at each step an $\bigoh(\nodes^{2})$ shortest path algorithm, leading to updates of significantly higher complexity.

\paragraph{Relation to the replicator dynamics}
If we differentiate the network's state variable $x$ in \eqref{eq:BR}, a straightforward calculation leads to the evolutionary dynamics
\begin{equation}
\label{eq:BR-primal}
\dot x_{\path}
	= -\temp x_{\path} \bracks*{\modcost_{\path}(x) - \insum_{\pathalt\sim\path} x_{\pathalt} \modcost_{\pathalt}(x)}
	+ \frac{\dot\temp}{\temp} x_{\path} \bracks*{\log x_{\path} - \insum_{\pathalt\sim\path} x_{\pathalt} \log x_{\pathalt}},
\end{equation}
where the summation is taken over all paths $\pathalt\sim\path$ with the same source as $\path$, and we have taken $\rate_{\source} = 1$ for simplicity.
Thus, in the baseline case $\temp = 1$, we obtain the dynamical system
\begin{equation}
\label{eq:RD}
\tag{RD}
\dot x_{\path}
	= -x_{\path} \bracks*{\modcost_{\path}(x) - \insum_{\pathalt\sim\path} x_{\pathalt} \modcost_{\pathalt}(x)},
\end{equation}
which is the classical (multi-population) \emph{replicator equation} of evolutionary game theory \cite{TJ78,San10}.
In this way, \eqref{eq:BR} can be seen as an extension of the replicator-based routing schemes of \cite{FV04,BEL06,KDB14,BK03}, the key difference being the extra, price-independent term of the dynamics \eqref{eq:BR-primal}.
As we shall see, this additional term will be crucial for the convergence properties of \eqref{eq:BR} under uncertainty \textendash\ an open problem posed by \cite{KDB14}.

\paragraph{On the inverse temperature $\temp$}
From an algorithmic viewpoint, the role of the inverse temperature parameter $\temp(t)$ in \eqref{eq:BR} is to act as an extrinsic weight that ``normalizes'' the paths' scores $\score(t)$.
In particular, for low $\temp$,
the Boltzmann distribution \eqref{eq:Boltzmann} tends to select paths uniformly;
by contrast,
for high $\temp$,
the induced choice probabilities ``freeze'' down to a hard best-response scheme which routes all traffic along the path with the lowest score.
Thus, to counterbalance the $\bigoh(t)$ growth of the variables $\score_{\path}(t)$,
we will assume throughout that $\temp(t)$ is \emph{nonincreasing} with decay rate slower than $1/t$:

\begin{assumption}
\label{asm:temp}
$\temp(t)$ is $C^{1}$-smooth, nonincreasing, and $\lim_{t\to\infty} \temp(t)t = \infty$.
\end{assumption}

In thermodynamic terms, Assumption \ref{asm:temp} means that the system is being \emph{heated} over time (instead of being cooled).
This comes in stark contrast with simulated annealing and log-linear learning \cite{Blu93} where the system begins at a high sampling temperature ($\temp\approx 0$) and subsequently freezes to very low temperatures ($\temp\to\infty$) to approach a state of least energy.
The reason for taking a heating schedule in \eqref{eq:BR} is that the energy levels $\score_{\path}(t)$ of the system at hand are not fixed (as in simulated annealing) but, instead, they grow over time.
As a result, if the paths' prices (whose aggregation determines the system's energy levels) are subject to randomness, freezing the system may lead it to quench prematurely to a suboptimal state;
we explore this issue in detail in Section \ref{sec:analysis}.

\subsection{Routing under uncertainty}
\label{sec:stochastic}

A key assumption underlying the (deterministic) Boltzmann routing scheme \eqref{eq:BR} is that prices are assumed immune to exogenous stochastic fluctuations.
In practice however, this assumption often fails:
background traffic fluctuations can be quite substantial due to burst-like user demands,
open market energy costs are highly volatile (for instance, at certain locations power may be supplied by renewable energy sources whose output depends on the weather),
whereas
load measurements are typically subject to errors due to delay \cite{Zhou2002_Study}, inaccurate routing information and noise \cite{Masip2003_Routing},
etc.

Starting from \eqref{eq:modcost-path}, we will model such disturbances via the random perturbation model
\begin{equation}
\label{eq:cost-stoch}
\tilde\modcost_{\path}(x;t)
	= \sum_{\edge\in\path} \parens[\big]{\modcost_{\edge}(w_{\edge}) + \xi_{\edge}(t)},
\end{equation}
where $\xi_{\edge}$ is a zero-mean stochastic process.%
\footnote{At the most basic level, $\xi_{\edge}$ can be assumed to be a simple \ac{AWGN} process.}
The score $\score_{\path}$ of path $\path$ will then follow the dynamics
\begin{equation}
\label{eq:BR-Langevin}
\begin{aligned}
\dot \score_{\path}
	&= \tilde\modcost_{\path}(x)
	= \modcost_{\path}(x) + \xi_{\path},
\end{aligned}
\end{equation}
where $\xi_{\path} = \sum_{\edge\in\alpha} \xi_{\edge} = \sum_{\edge\in\edges} \pmat_{\edge\path} \xi_{\edge}$ denotes the aggregate perturbation over path $\path$.
Hence, writing \eqref{eq:BR-Langevin} as an Itô (non-anticipative) \acl{SDE} \cite{Oks07,Kuo06,Dur96}, we obtain the \emph{stochastic Boltzmann routing} dynamics
\begin{equation}
\label{eq:SBR}
\tag{SBR}
\begin{aligned}
dY_{\path}
	&= \modcost_{\alpha}(X) \dd t + \dd Z_{\path},
	\\
X_{\path}
	&= \frac{\rate_{\source} \expof{-\temp Y_{\path}}}{\sum_{\pathalt\in\paths_{\source}} \expof{-\temp Y_{\pathalt}}}.
\end{aligned}
\end{equation}
where, to be consistent with \eqref{eq:BR-Langevin}, the \emph{path-noise process} $Z_{\path}$ is of the form
\begin{equation}
\label{eq:noise}
dZ_{\path}
	= \sum_{\edge\in\path} \noisedev_{\edge} \dd W_{\edge},
\end{equation}
where $W = (W_{\edge})_{\edge\in\edges}$ is a standard Wiener process (Brownian motion) in $\R^{\edges}$ and $\noisedev_{\edge} \equiv \noisedev_{\edge}(x,t)$ is the (possibly state- and time-dependent) \emph{volatility coefficient} of the fluctuations on edge $\edge\in\edges$.

The dynamical system \eqref{eq:SBR} will be our core stochastic model so we proceed with some explanatory remarks:

\paragraph{Assumptions on the noise}
The noise model \eqref{eq:noise} implies that fluctuations are independent across \emph{links}, but not across \emph{paths}.
Indeed, a simple calculation reveals the correlation structure
\begin{equation}
\label{eq:corr}
dZ_{\path} \cdot dZ_{\pathalt}
	= \sum_{\edge,\edgealt\in\edges} \noisedev_{\edge} \noisedev_{\edge'} \dd W_{\edge} \cdot dW_{\edgealt}
	= \sum_{\edge,\edgealt\in\edges} \noisedev_{\edge} \noisedev_{\edge'} \delta_{\edge\edgealt} \dd t
	= \sum_{\edge\in\path\cap\pathalt} \noisevar_{\edge} \dd t,
\end{equation}
i.e. \emph{fluctuations along two paths are correlated along their common edges}.%
\footnote{More generally, the edge processes $W_{\edge}$ could be themselves correlated along different edges;
we do not consider this for simplicity.}
With this in mind, it will be convenient to introduce the \emph{volatility matrix} $\covmat = (\noisevar_{\path\pathalt})_{\path,\pathalt\in\paths}^{\phantom{2}}$ defined as
\begin{equation}
\label{eq:covmat}
\noisevar_{\path\pathalt}
	\equiv \sum_{\edge\in\path\cap\pathalt} \noisevar_{\edge}
	= \sum_{\edge\in\edges} \pmat_{\edge\path} \pmat_{\edge\pathalt} \noisevar_{\edge}.
\end{equation}
By construction, the matrix $\covmat \equiv \covmat(x,t)$ describes the \emph{quadratic covariation} of the noise process $Z$ in the sense that $dZ_{\path} \cdot dZ_{\pathalt} = \covmat_{\path\pathalt} \dd t$ \cite{Kuo06,Oks07}.
On that account, our only assumption will be:

\begin{assumption}
\label{asm:noise}
The fluctuations' volatility matrix $\covmat$ is bounded:
\(
\sup_{x,t} \norm{\covmat(x,t)}
	\equiv \noisebound^{2}
	< \infty.
\)
\end{assumption}

\begin{remark}
We should clarify here that Assumption \ref{asm:noise} means that fluctuations are only bounded in \emph{mean square;}
at any given time $t\geq0$, the disturbances $dZ_{\edge}(t)$ could be arbitrarily large.
\end{remark}


\paragraph{The stochastic replicator dynamics}
Even though $X(t)$ is fully determined under \eqref{eq:SBR} for all $t\geq0$, this is an indirect description of the governing dynamics of $X(t)$.
To obtain an explicit description of these (stochastic) dynamics,
we can follow the same approach as in \eqref{eq:BR-primal};
however, because of the noise, we now have to employ the rules of (Itô) stochastic calculus \cite{Dur96,Kuo06,Oks07}.
This derivation is quite convoluted, so we only present here the end result (proven in Appendix \ref{app:basics}):

\begin{proposition}
\label{prop:evolution}
The solutions of \eqref{eq:SBR} satisfy the \acl{SDE}:
\begin{subequations}
\label{eq:SRD}
\begin{flalign}
dX_{\path}
	=
	&\label{eq:SRD-det}
	-\temp X_{\path}
	\left[
	\modcost_{\path}(X) - \insum_{\pathalt} X_{\pathalt}\,\modcost_{\pathalt}(X)
	\right] dt
	\\
	&\label{eq:SRD-noise}
	- \temp X_{\path}
	\left[
	dZ_{\path} - \insum_{\pathalt} X_{\pathalt} \dd Z_{\pathalt}
	\right]
	\\
	&\label{eq:SRD-temp}
	+\frac{\dot \temp}{\temp} X_{\path}
	\left[
	\log X_{\path} - \insum_{\pathalt} X_{\pathalt} \log X_{\pathalt}
	\right] dt
	\\
	&\label{eq:SRD-Ito}
	+\frac{\temp^{2}}{2} X_{\path}
	\left[
	\insum_{\pathalt,\gamma} (\delta_{\path\gamma} - X_{\gamma}) \, (\delta_{\pathalt\gamma} - 2 X_{\pathalt}) \, \noisevar_{\pathalt\gamma}
	\right] dt,
\end{flalign}
\end{subequations}
where
all sums are taken over the paths $\pathalt,\gamma\sim\path$ with the same starting point as $\path$,
and
all traffic generation rates have been taken equal to $\rate_{\source} = 1$ for simplicity.
\end{proposition}

Despite the complex appearance of \eqref{eq:SRD}, each of the constituent terms admits a relatively straightforward interpretation:
\begin{enumerate}
[\itshape a\upshape)]
\item
The term \eqref{eq:SRD-det} drives the process in the baseline, deterministic case $\noisedev = 0$, $\temp = \textrm{constant}$;
as such, it coincides with the deterministic replicator dynamics \eqref{eq:RD}.
\item
The martingale term \eqref{eq:SRD-noise} reflects the direct impact of the noise on the evolution of $X(t)$.
\item
The term \eqref{eq:SRD-temp} is due to the temporal variation of the inverse temperature parameter $\temp(t)$ and plays the same role as in \eqref{eq:BR-primal}.
\item
Finally, the term \eqref{eq:SRD-Ito} is the Itô correction induced by the non-anticipative nature of the Itô integral \cite{Kuo06,Oks07}.
This term is \emph{price-independent} and does not depend on $\dot\temp$, so it does not vanish for constant $\temp$;
also, the correlation structure of the noise process $Z$ appears explicitly in \eqref{eq:SRD-Ito} via the volatility matrix $\covmat$ (cf. the remarks preceding Assumption \ref{asm:noise}).
\end{enumerate}

Apart from a vaguely similar structure, there is no overlap between \eqref{eq:SRD} and the stochastic replicator dynamics with ``aggregate shocks'' that have been studied extensively in evolutionary biology \cite{FH92,Cab00,Imh05,HI09,MV16}.
The exponential learning approach of \cite{MM10,BM14} is much closer in spirit but it does not account for the nonlinear nature of the cost functions $\modcost_{\edge}$ and the correlation structure of the noise processes $Z$ along paths;
we explore these connections in more detail below.

%% file: Analysis.tex

Our aim in this section will be to analyze the long-term convergence properties of the proposed Boltzmann routing scheme in the presence of noise and uncertainty.
To establish a baseline, we begin with the scheme's convergence properties in a noiseless, deterministic setting:

\begin{theorem}
\label{thm:conv-det}
Under \eqref{eq:BR}, $x(t)$ converges to a solution of the global minimization problem \eqref{eq:CM}.
\end{theorem}

Given the connection of \eqref{eq:BR} to the replicator dynamics, Theorem \ref{thm:conv-det} (proven in Appendix~\ref{app:deterministic}) represents a strengthening of known convergence results for potential \cite{San01} and stable games \cite{HS09}.
Specifically, \cite{San01} showed that the replicator dynamics \eqref{eq:RD} converge to the set of rest points in potential games;
however, since the rest points of \eqref{eq:RD} are not necessarily solutions of \eqref{eq:CM}, we cannot use the analysis of \cite{San01,San10} to deduce the convergence of \eqref{eq:BR} to a globally efficient state.
More recently, \cite{FV04} established the convergence of \eqref{eq:RD} to the game's equilibrium set, but under the implicit assumption that the network's routes are linearly independent (an assumption that fails automatically if $\abs{\paths} > \abs{\edges}$ because $\pmat$ fails to be invertible on $\strat$ in this case).
Instead, Theorem \ref{thm:conv-det} dispenses with all such assumptions and provides an unconditional convergence result for \eqref{eq:BR}.

This behavior of \eqref{eq:BR} is fairly encouraging, but it is contingent on the absence of fluctuations and uncertainty.
If the network is constantly subject to stochastic disturbances, there is no reason to expect that this convergence still holds beyond the ``small noise'' regime.
Nevertheless, as we show below, the proposed Boltzmann routing scheme remains exceptionally robust in the presence of noise:
under \eqref{eq:SBR}, the long-term average of the total network consumption is minimized, \emph{irrespective of the level of uncertainty}.
Formally:

\begin{theorem}
\label{thm:averages}
Assume that \eqref{eq:SBR} is run with a variable parameter $\temp(t)$ satisfying Assumption \ref{asm:temp}.
Then, the long-term average $\bar\obj(t) = t^{-1} \int_{0}^{t} \obj(X(s)) \dd s$ of the network's total consumption enjoys the bound:
\begin{equation}
\label{eq:cost-bound}
\bar\obj(t)
	\leq \obj^{\ast}
	+ \frac{\sum_{\source}\log\nPaths_{\source}}{\temp(t) t}
	+ \frac{\noisebound^{2}}{2 t} \int_{0}^{t} \temp(s) \dd s
	+ 2 \noisebound^{2} \sqrt{\frac{\log\log t}{t}}
	+ \bigoh(1/t)
	\quad
	\as
\end{equation}
where $\obj^{\ast}$ is the minimum value of \eqref{eq:CM} and $\nPaths_{\source} = \abs{\paths_{\source}}$ is the number of paths utilized by source $\source\in\sources$.
In particular, if $\temp(t)\to0$, we have $\lim_{t\to\infty} \bar\obj(t) = \obj^{\ast}$ \as.
\end{theorem}

The proof of Theorem \ref{thm:averages} (presented in detail in Appendix \ref{app:averages}) hinges on the so-called
\acdef{KL} divergence \cite{KL51,CT91,MS16}, defined here as
\begin{equation}
\label{eq:KL}
\dkl(\eq,x)
	= \insum_{\source\in\sources} \insum_{\path\in\paths_{\source}} \eq_{\alpha} \log (\eq_{\alpha} / x_{\alpha}),
	\quad
	\eq,x\in\feas.
\end{equation}
The \ac{KL} divergence provides an asymmetric measure of the information-theoretic distance between $\eq$ and $x$.
Using this, we express the average global consumption $\bar\obj(t)$ in terms of this distance, and we then use Itô's lemma \cite{Oks07,Kuo06,Dur96} to bound it by a vanishing function of $t$.
In so doing, we obtain the three main components of the convergence rate estimate \eqref{eq:cost-bound}:
the first represents the convergence rate of the noiseless process \eqref{eq:BR},
the second is due to the Itô correction \eqref{eq:SRD-Ito},
while the third one stems from the law of the iterated logarithm \cite{Dur96}.

Importantly, the second term of \eqref{eq:cost-bound} does not vanish for constant $\temp$, explaining the requirement $\lim_{t\to\infty} \temp(t) = 0$.
On the other hand, this requirement can be dropped if there is no noise:
for $\noisebound = 0$ and $\temp = \text{constant}$, the RHS of \eqref{eq:cost-bound} reduces to $\bigoh(1/t)$.
In addition, the form of \eqref{eq:cost-bound} also highlights a trade-off between more aggressive (slowly decaying) schedules for $\temp$ and the underlying uncertainty.
Specifically, the first (deterministic) term of \eqref{eq:cost-bound} is decreasing in $\temp$ while the second (stochastic) one increases with $\temp$.
Hence, in the absence of noise, it is better to use a large, constant $\temp$ instead of letting $\temp(t)\to0$;
however, a constant $\temp$ may be detrimental under uncertainty.

These considerations can be illustrated by choosing a specific heating schedule of the form $\temp(t) \sim 1/t^{a}$ for some $a\in(0,1)$;
in this case, we obtain the explicit convergence rate:
\begin{corollary}
\label{cor:averages}
Assume that \eqref{eq:SBR} is run with $\temp(t) \sim 1/t^{a}$ for some $a\in(0,1)$.
Then:
\begin{equation}
\label{eq:rates}
\bar\obj(t)
	= \obj^{\ast}
	+ \begin{cases}
		\bigoh(1/t^{a})
			&\quad
			\text{if $0<a<\frac{1}{2}$},
			\\
		\bigoh\parens*{\sqrt{\log \log t / t}}
			&\quad
			\text{if $a = 1/2$},
			\\
		\bigoh(1/t^{1-a})
			&\quad
			\text{if $\tfrac{1}{2} < a <1$}.
		\end{cases}
\end{equation}
\end{corollary}


\begin{IEEEproof}
Simply note that $t^{\max\{a-1,-a\}}$ is the dominant term in \eqref{eq:cost-bound} for all $a\neq 1/2$.
Otherwise, for $a=1/2$, the first two terms of \eqref{eq:cost-bound} are both $\bigoh(t^{-1/2})$ and are dominated by the third.
\end{IEEEproof}

Theorem \ref{thm:averages} and Corollary \ref{cor:averages} make no assumptions for the underlying network or the magnitude of the stochastic perturbations affecting the system.
However, the convergence they provide is in terms of the long-term average consumption $\bar\obj(t)$, not the instantaneous network state $X(t)$.
As a matter of fact, in the presence of disturbances, $X(t)$ cannot converge with positive probaibility to a fully-mixing state $\eq$ where a given source node utilizes several paths concurrently:
even if $\eq$ is a rest point of \eqref{eq:BR}, the noise term \eqref{eq:SRD-noise} does not vanish at $\eq$, so $\eq$ cannot be stationary under \eqref{eq:SBR}.
On the other hand, if \eqref{eq:CM} admits a \emph{strict, non-mixing} solution $\eq$ (so each source node utilizes a \emph{single} path at \acl{NE}; cf. Definition \ref{def:Nash}), we show below that Boltzmann routing converges to $\eq$ \emph{independently of the magnitude of the noise:}

\begin{theorem}
\label{thm:strict}
Assume that the global minimization problem \eqref{eq:CM} admits a strict, non-mixing solution $\eq$.
If \eqref{eq:SBR} is run with sufficiently small \textup(constant\textup) $\temp$, $X(t)$ converges to $\eq$ \as.
\end{theorem}


\begin{remark}
We should note here that if $\eq$ is a strict, non-mixing solution of \eqref{eq:CM}, it is the \emph{only} solution of \eqref{eq:CM};
as such, Theorem \ref{thm:strict} shows that Boltzmann routing converges \as to the network's \emph{unique} globally efficient state.
\end{remark}

The proof of Theorem \ref{thm:strict} is fairly complicated, so we relegate it to Appendix \ref{app:strict}.
In a nutshell, it consists of showing that
\begin{inparaenum}
[\itshape a\upshape)]
\item
if $X(t)$ remains close to $\eq$ for all time, it is eventually attracted to it;
and
\item
$X(t)$ gets trapped in arbitrarily small neighborhoods of $\eq$ with controllably high probability (a much harder result which relies on an application of Girsanov’s theorem \cite{Kuo06,Oks07} to estimate the probability that a Wiener process with positive drift attains a given negative level in finite time).
\end{inparaenum}

From a practical viewpoint, the importance of Theorem \ref{thm:strict} is that it provides a global convergence result for constant (but small) $\temp$, irrespective of the noise level.
This relaxes even further the ``vanishing $\temp$'' requirement of Theorem \ref{thm:averages} and is owed to the existence of a \emph{strict}, non-mixing solution.
On the flip side, if the network does not admit such a solution, $X(t)$ cannot converge with positive probability \textendash\ even though the long-term average consumption $\bar\obj(t)$ does.
Our next result shows that if the network admits an (isolated) interior solution $\eq$, learning with sufficiently small $\temp$ allows $X(t)$ to remain arbitrarily close to $\eq$ with arbitrarily high probability:

\begin{theorem}
\label{thm:interior}
Let $\eq\in\intr(\feas)$ be an isolated, interior solution of \eqref{eq:CM}.
Then, for all $\eps,\delta>0$,
the scheme \eqref{eq:SBR} can be run with a sufficiently small \textup(constant\textup) parameter $\temp \equiv \temp(\eps,\delta)$ such that
\begin{equation}
\label{eq:invariant}
\probof[\big]{\textup{$\norm{X(t) - \eq} \leq \delta$ for all sufficiently large $t$}}
	\geq 1 - \eps.
\end{equation}
\end{theorem}

Theorem \ref{thm:interior} implies that if the inverse temperature $\temp$ of \eqref{eq:SBR} is small enough, then, after a certain amount of time, the system will be arbitrarily close to equilibrium with probability arbitrarily close to one.
Thus, even though the network's state does not converge almost surely, it will still spend most of the time close to equilibrium.

%% file: Numerics.tex

In this section, we assess the performance of the proposed Boltzmann routing scheme in practical scenarios via numerical simulations.
For clarity, we only present here a representative subset of these results but our conclusions apply to a wide range of optical network parameters and specifications.

Our setup is as follows (for an overview, see Table \ref{tab:parameters}):
we consider an optical network deployed over the continental US with nodes and links as in Fig.~\ref{fig:network} (for clarity, we only plotted the 50 largest metropolitan areas) \cite{DRC10}.
At each link, data is transmitted over 80 \ac{WDM} channels spaced at $50\;\ghz$, each with a carrying capacity of $10\;\gbs$.
To mitigate signal losses, each link carries an amplifier per $80\;\km$ of fiber length, and an additional amplifier at each end;
in addition to the amplifiers, power is also consumed at the transponder and switch port levels (cf. Table~\ref{tab:parameters}) \cite{BGT+13}.
Source nodes are drawn randomly in the network, each with a traffic rate between $200\;\gbs$ and $400\;\gbs$;
this traffic is then directed to a set of \aclp{DC}, also drawn randomly from the network's nodes.
Each \acl{DC} has a capacity of $2\;\tbs$ and consumes between $6.6\;\kW$ and $13.2\;\kW$, at zero and full load respectively \cite{FWB07};
for diversity, we also populate the network with a number of ``legacy'' \aclp{DC} with lower capacity and higher power consumption specifications as indicated in \cite{Pul11}.

\begin{table}[t]
\caption{Network simulation parameters}
\label{tab:parameters}
\vspace{-2ex}
\centering
\footnotesize
\input{Parameters}
\end{table}

In Fig.~\ref{fig:statics}, we examine the performance of the proposed Boltzmann routing scheme in terms of power consumption, convergence speed and scalability.
Specifically, we consider several different structures for the source nodes' choice sets $\paths_{\source}$:
\begin{inparaenum}
[\itshape a\upshape)]
\item
taking the shortest path to the closest destination ($\abs{\act_{\source}} = 1$);
\item
splitting traffic over the $4$ paths with the lowest hop count to the closest destination ($\abs{\act_{\source}} = 4$);
\item
splitting traffic to the $5$ closest destinations over shortest paths ($\abs{\act_{\source}} = 5$);
and
\item
mixing shortest paths and closest destinations ($\abs{\act_{\source}} = 4\times5$).
\end{inparaenum}
We then run the Boltzmann routing scheme \eqref{eq:BR} with the Pigouvian pricing scheme of Proposition \ref{prop:Pigou} ($\modcost_{\edge} = \cost_{\edge}'$), a heating schedule of the form $\temp(t)\sim t^{-1/2}$ and all scores initialized uniformly at $Y(0)=0$ (of course, in the single-path, single-destination case, there is nothing to learn).

\begin{figure*}[t]
\footnotesize
\subfigure[Power gains under different routing modes]{\label{fig:modes}%
\includegraphics[width=.48\textwidth]{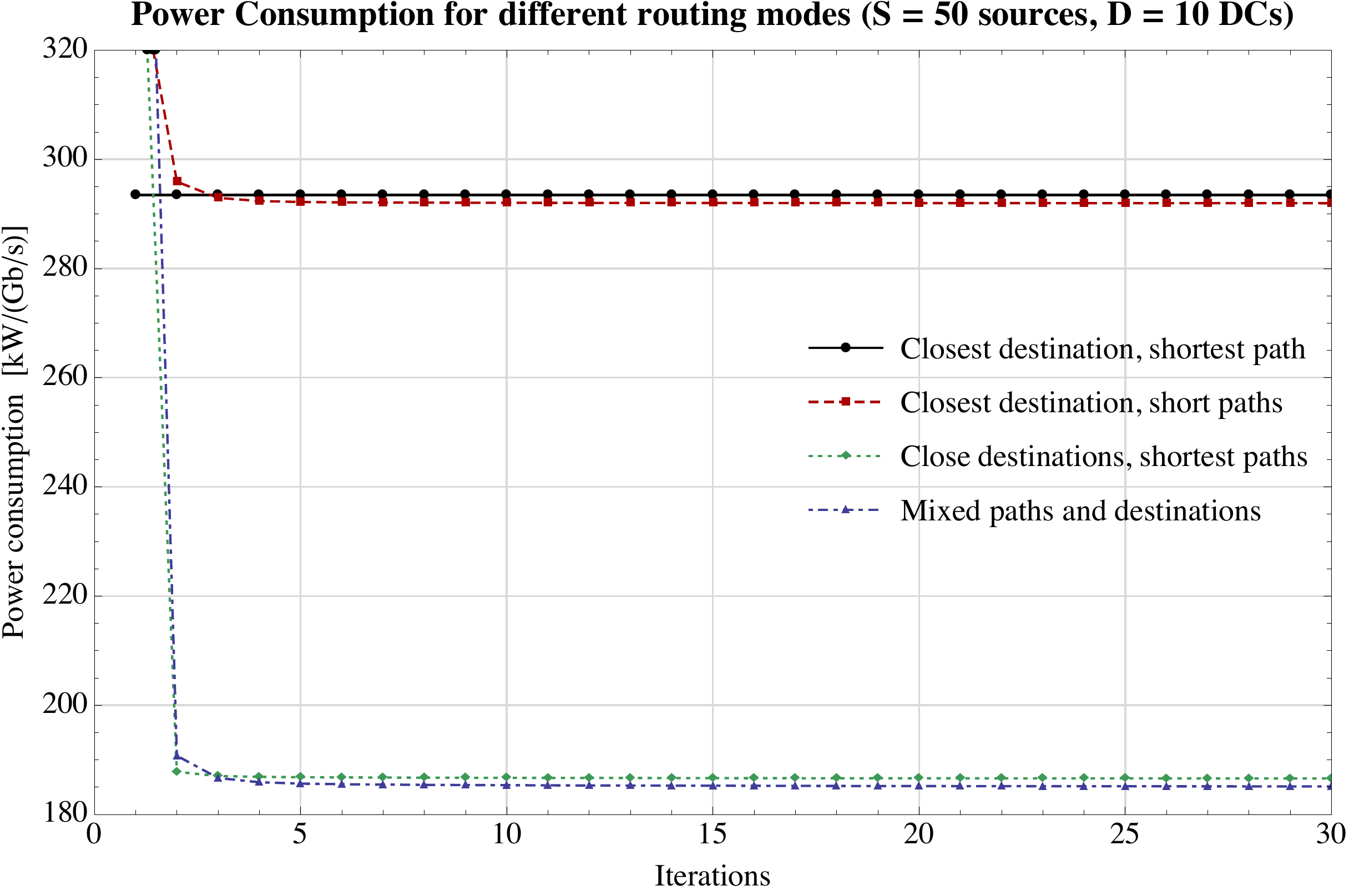}}
\hfill
\subfigure[Scalability of Boltzmann routing]{\label{fig:scalability}%
\includegraphics[width=.48\textwidth]{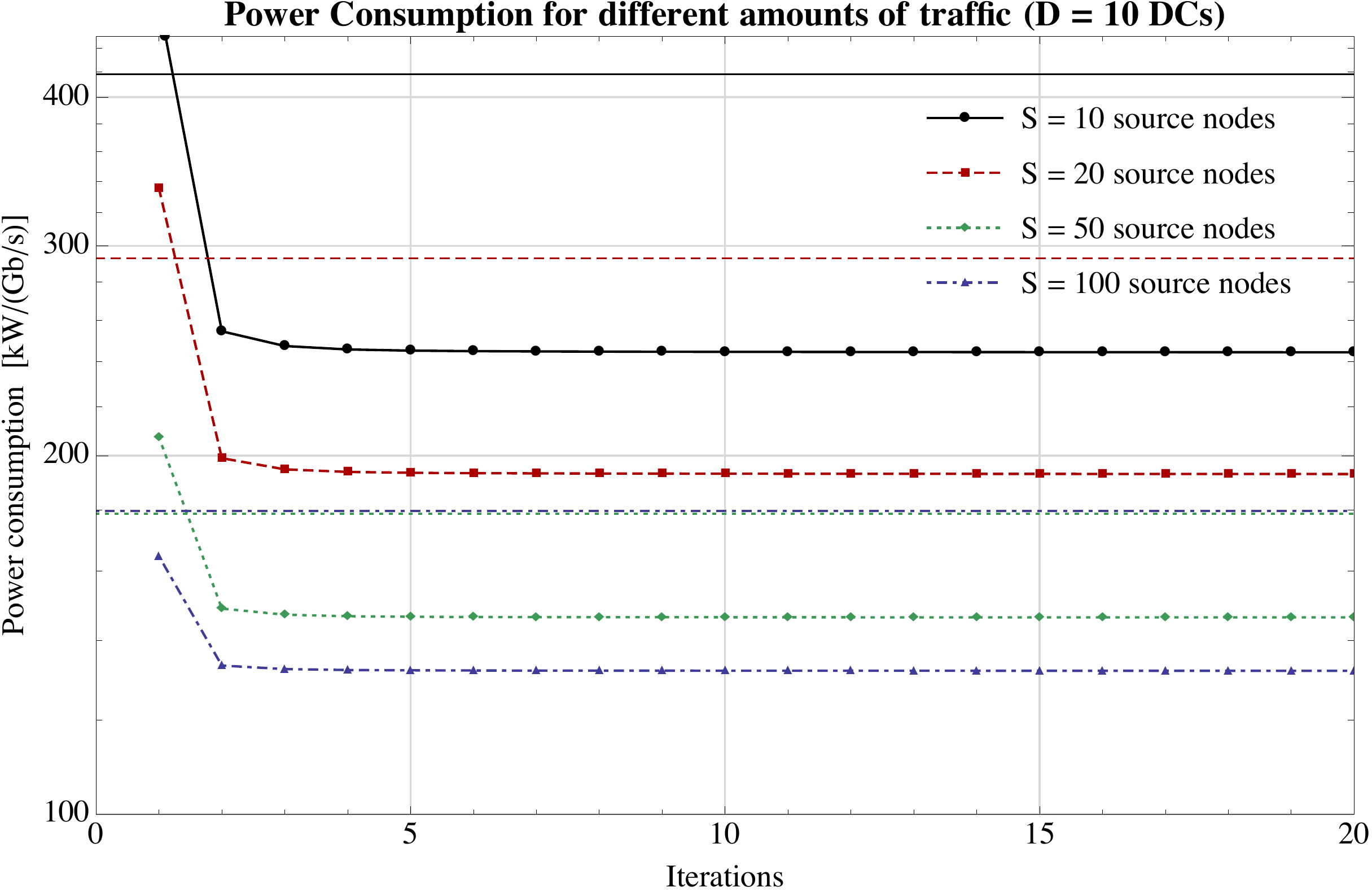}}
\\[-4ex]
\caption{Gains in power consumption for different routing modes and different numbers of sources.
Fig.~\ref{fig:modes} shows that splitting traffic to several destinations leads to a $100\;\kW/(\gbs)$ decrease in power consumption.
As we see in Fig.~\ref{fig:scalability}, these gains become more pronounced for larger number of source nodes (marked/unmarked lines correspond to mixed/closest-destination routing respectively); 
moreover, Boltzmann routing attains $99\%$ of its overall gain within 4\textendash5 iterations, even for $\nSources = 100$ source nodes.%
}
\vspace{-1ex}
\label{fig:statics}
\end{figure*}

As shown in Fig.~\ref{fig:modes}, Boltzmann routing converges very rapidly (within a few iterations) to a socially efficient network state that solves the corresponding consumption minimization problem \eqref{eq:CM}.
Moreover, we see that splitting traffic over several \aclp{DC} (as opposed to exclusively targeting the closest \acl{DC}) leads to significant performance gains, of the order of $100\;\kW/(\gbs)$ \textendash\ approximately a $40\%$ decrease in power consumption with respect to closest-destination routing.
Finally, in Fig.~\ref{fig:scalability}, we examine the scalability of these performance gains as the number of source nodes increases:
in so doing, we observe that
\begin{inparaenum}
[\itshape a\upshape)]
\item
the relative percentage gain of splitting traffic over multiple \aclp{DC} becomes more prominent for large numbers of sources;
and
\item
the Boltzmann routing scheme retains its convergence speed and attains a socially efficient state within 4\textendash 5 iterations, even for $\nSources=100$ source nodes.
This means that, in the presence of intermittent traffic fluctuations, the system will be capable of adapting fast enough to changes in the system.
\end{inparaenum}

\begin{figure*}[t]
\footnotesize
\subfigure[Convergence in the presence of disturbances]{\label{fig:noise}%
\includegraphics[width=.48\textwidth]{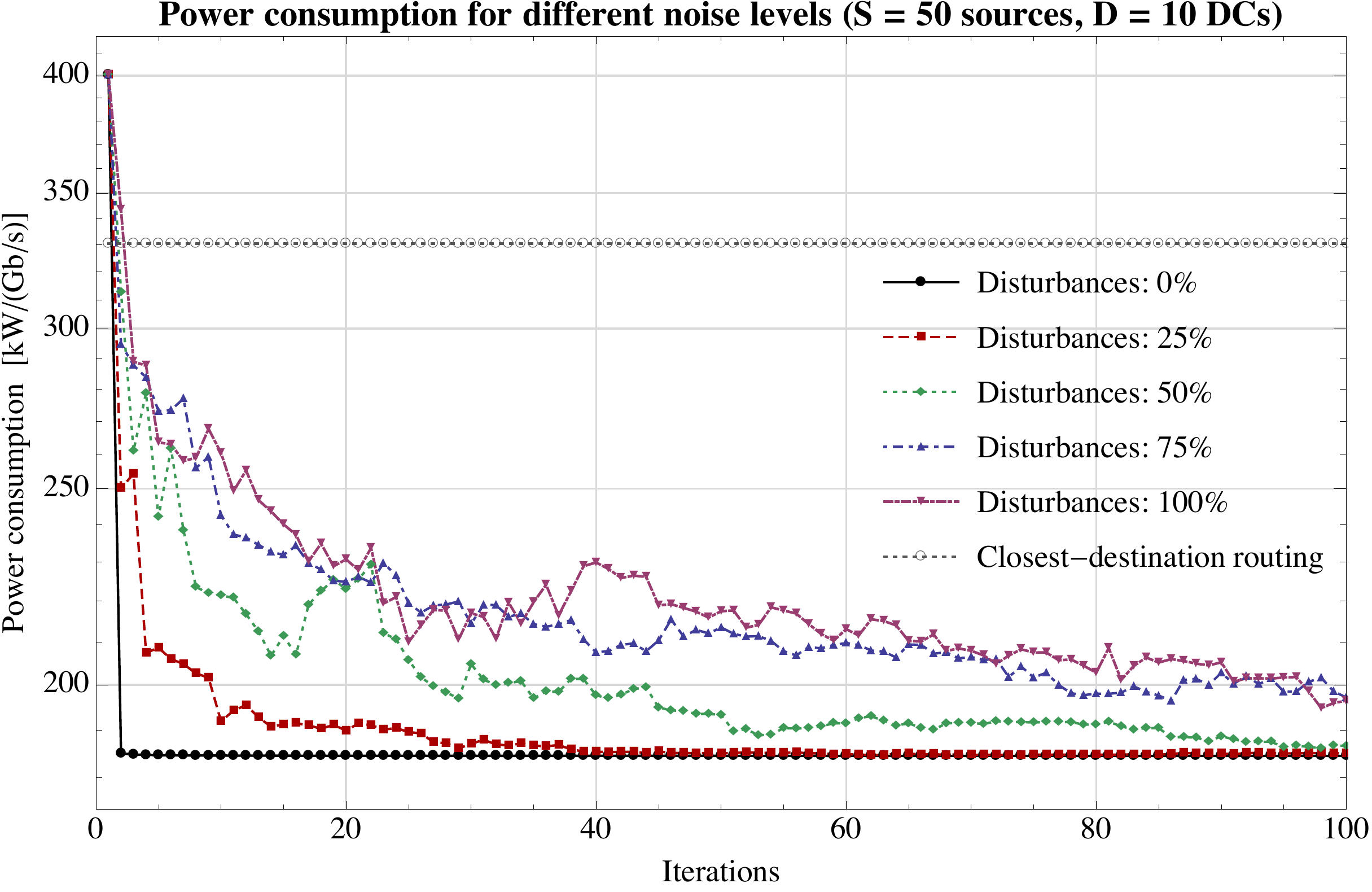}}
\hfill
\subfigure[Capacity violations]{\label{fig:capacities}%
\includegraphics[width=.48\textwidth]{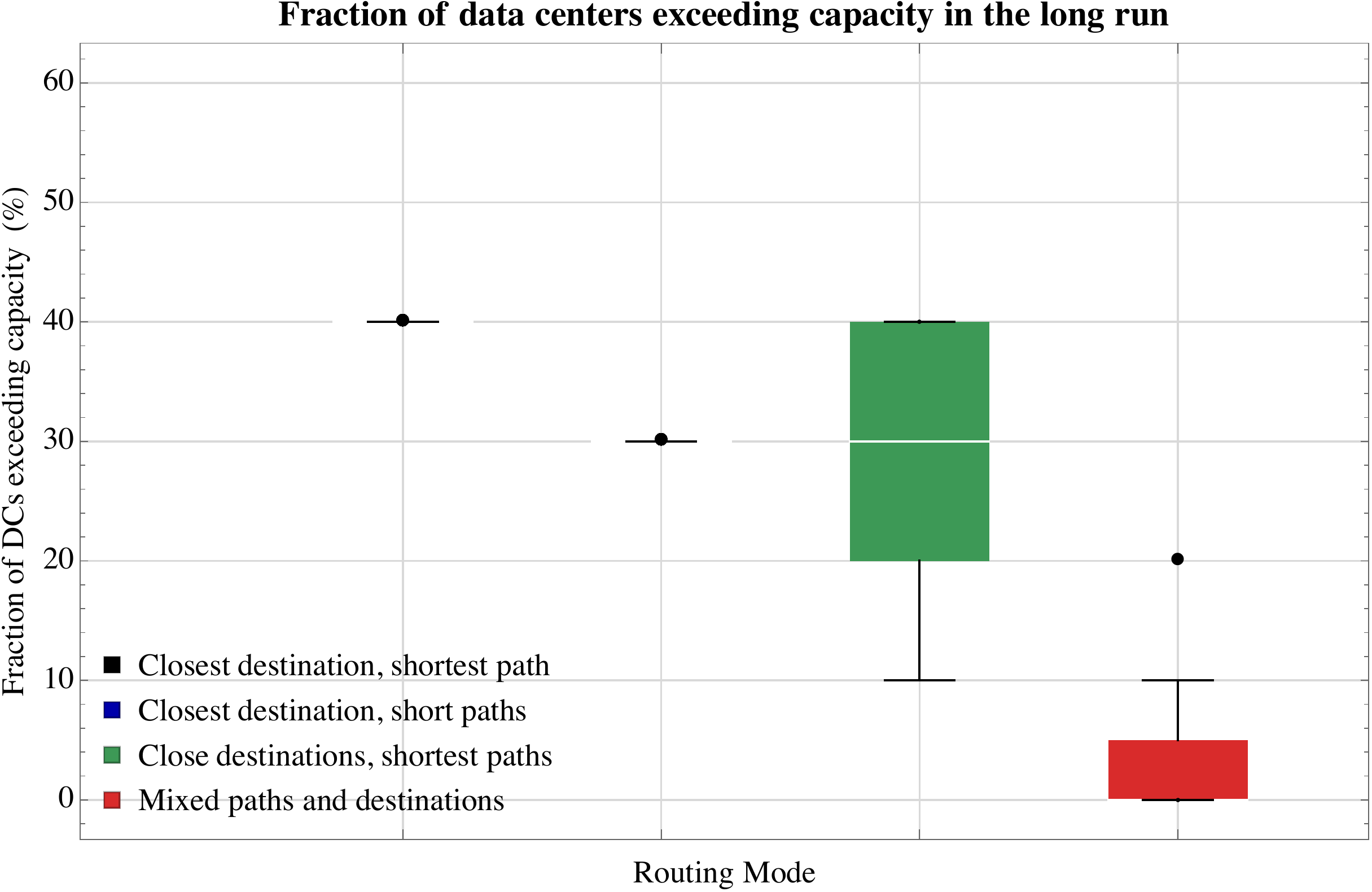}}
\\[-4ex]
\caption{Long-term behavior of Boltzmann routing in the presence of disturbances (Fig.~\ref{fig:noise}) and capacity constraints (Fig.~\ref{fig:capacities}).}
\label{fig:robustness}
\vspace{-2ex}
\end{figure*}

In Fig.~\ref{fig:robustness}, we examine the algorithm's robustness in the presence of stochastic fluctuations and capacity constraints.
First, in Fig.~\ref{fig:noise}, we consider a network with $\nSources = 50$ source nodes and $\nSinks = 10$ \aclp{DC} with power consumption and capacity attributes as in Table \ref{tab:parameters}.
We further assume that consumption costs are subject to stochastic disturbances and noise with volatility coefficient $\noisebound$ equal to a fraction $z$ of the costs' mean value (ranging between $0\%$ and $100\%$).
Remarkably, Boltzmann routing reaches a socially efficient state even under high degrees of uncertainty (up to $z=100\%$);
however, as could be expected, the algorithm's convergence speed decreases in the presence of large disturbances.

Finally, in Fig.~\ref{fig:capacities}, we examine the behavior of Boltzmann routing with respect to the capacity constraints imposed by the network's links and \aclp{DC}.
Specifically, we run \eqref{eq:SBR} for different routing modes (cf. Fig.~\ref{fig:statics}) with a $25\%$ disturbance level and we record the number of links and/or \aclp{DC} that exceed capacity at any given time.
The results are plotted in the box-and-whisker charts of Fig.~\ref{fig:capacities} which show that routing to a single destination (over a single or several paths) consistently leads to relatively high violation percentages, between $30\%$ and $40\%$.
On the other hand, except for a few outliers, splitting traffic over several paths and destinations yields a much lower capacity violation percentage, of the order of $5\%$ (a result which is statistically significant within 5 interquartile ranges).
This behavior stems from the fact that the Pigouvian pricing scheme of Proposition \ref{prop:Pigou} also penalizes users when a link (or \acl{DC}) approaches capacity because the slope of the $\eps$-adjusted cost function $\cost_{\edge}^{\eps}$ is equal to $1/\eps$ in that case (in our simulations, we took $\eps=10^{-4}$; recall also the relevant discussion in Section \ref{sec:consumption}).

%% file: Parameters.tex

\begin{tabular}{|c|c||c|c|}
\hline
\textbf{Parameter}
	&\textbf{Value}
	&\textbf{Parameter}
	&\textbf{Value}
	\\
	\hline
Network topology
	&Continental US \cite{DRC10}
	&Source traffic
	&$[200,400]\;\gbs$
\\
\hline
\ac{WDM} channels\,/\,fiber
	&$80$
	&Channel capacity
	&$10\;\gbs$
\\
\hline
Amplifier consumption
	&$15\;\W$
	&Amplifier density
	&$1 / (80\;\km)$
\\
\hline
Transponder consumption
	&$35\;\W/\mathrm{chanel}$
	&Switch port consumption
	&$0.8\;\W$
\\
\hline
\ac{DC} consumption (zero load)
	&$6.6\;\kW$
	&\ac{DC} consumption (full load)
	&$13.2\;\kW$
\\
\hline
\ac{DC} capacity
	&$2\;\tbs$
	&Cooling factor
	&$2$
\\
\hline
\end{tabular}

%% file: Conclusions.tex

The prolific increase of traffic in optical networks is forcing network operators to seek energy-efficient routing methods while facing capacity limits and random fluctuations.
This paper is an attempt to propose such a scheme:
specifically, we analyzed a distributed approach for energy-efficient routing in optical data networks based on the implementation of a Pigouvian pricing scheme and a learning method derived from the Boltzmann distribution of statistical mechanics.
The \aclp{NE} of the proposed pricing scheme coincide with the network's socially optimum states (in terms of total energy consumption), and the resulting Boltzmann routing method converges to such states exponentially fast in realistic network conditions, leading to gains of up to $40\%$ in energy consumption over simple, shortest-path routing schemes.

Motivated by the highly volatile nature of traffic and power consumption in optical networks, we also examine the behavior of Boltzmann routing in the presence of uncertainty and random disturbances.
In this regard, we showed that the network's long-term average consumption converges within $\eps$ of its minimum value in time which is at most $\tilde\bigoh(\eps^{-2})$, irrespective of the fluctuations' magnitude.
Moreover, if the network admits a strict, non-mixing optimum state, the network's state converges to it almost surely (again, no matter the noise level);
instead, if the network admits an interior, fully-mixing optimum state, the proposed scheme remains arbitrarily close to the said state  with probability arbitrarily close to $1$, provided that the scheme's inverse temperature parameter $\eta$ is taken sufficiently small.
These theoretical results were complemented by extensive numerical simulations showing that, in realistic network conditions, Boltzmann routing provides significant advantages over closest-destination/shortest-path routing choices.
In addition, the proposed method scales up to hundreds of source nodes with minimal impact on its convergence speed, and it respects the network's capacity constraints, even under very high degrees of volatility and uncertainty.

The results presented in this paper apply to a wide range of congestion-limited network scenarios, ranging from latency minimization and throughput maximization in data networks to urban traffic management in road networks;
the detailed analysis of these applications is relegated to future work.
Additionally, a key question that arises is what happens in the presence of non-zero-mean fluctuations (random or not), which could lead the system's optimum/equilibrium states to evolve over time in an arbitrary fashion (so there is no fixed underlying state to target).
In this case, efficient routing protocols would have to be capable of tracking the system's evolving optimum state in a flexible, dynamic manner;
we intend to explore this issue in future work.

%% file: App-Basics.tex

For simplicity, in what follows, we take all rates equal to $\rate_{\source} = 1$;
this assumption is only done for notational clarity and does not otherwise affect our results.
For posterity, define also the Gibbs map $y\mapsto\gibbs(y)$ with $\gibbs_{\path}(y) = \exp(y_{\path})/\insum_{\pathalt\sim\path} \exp(y_{\pathalt})$, so the update step of \eqref{eq:BR} can be written as $x = \gibbs(y)$.
Finally, with a fair degree of hindsight, introduce the (negative) entropy function
\(
h(x)
	= \sum_{\path\in\paths} x_{\path} \log x_{\path},
	\quad
	x\in\feas,
\)
and its convex conjugate \cite{Roc70}
\begin{equation}
\label{eq:conjugate}
h^{\ast}(y)
	= \max_{x\in\feas} \braces{\braket{y}{x} - h(x)}
	= \sum_{\source\in\sources} \log \sum_{\path\in\paths_{\source}} e^{y_{\path}},
	\quad
	y\in\R^{\paths}.
\end{equation}
Then, an easy differentiation shows that $G(y) = \nabla h^{\ast}(y)$ and, furthermore, for all $\pathalt,\gamma\sim\path$:
\begin{align}
\label{eq:G-diff}
\frac{\pd\gibbs_{\path}}{\pd x_{\pathalt}}
	= x_{\path} (\delta_{\path\pathalt} - x_{\pathalt}),
	&&
\frac{\pd^{2}\gibbs_{\path}}{\pd x_{\pathalt} \pd x_{\gamma}}
	= x_{\path} (\delta_{\path\pathalt} - x_{\pathalt}) (\delta_{\pathalt\gamma} - 2 x_{\pathalt}),
\end{align}
where $\delta_{\path\pathalt}$ denote Kronecker's delta symbols.
With these preliminaries at hand, we have:

\begin{IEEEproof}[Proof of Proposition \ref{prop:evolution}]
Applying Itô's formula \cite{Oks07,Kuo06} to the dynamics \eqref{eq:SBR}, we get:
\begin{equation}
dX_{\path}
	= \temp \sum_{\pathalt\sim\path}
	\left.\frac{\pd \gibbs_{\path}}{\pd x_{\pathalt}}\right\vert_{\temp Y} d Y_{\pathalt}
	+ \dot \temp \sum_{\pathalt\sim\path} Y_{\pathalt}
	\left.\frac{\pd \gibbs_{\path}}{\pd x_{\pathalt}}\right\vert_{\temp Y} dt
	+ \frac{\temp^{2}}{2} \sum_{\pathalt,\gamma\sim\path}
	\left.\frac{\pd^{2}\gibbs_{\path}}{\pd x_{\pathalt} \pd x_{\gamma}} \right\vert_{\temp Y}
	dY_{\pathalt} \cdot dY_{\gamma}.
\end{equation}
The system \eqref{eq:SRD} is then obtained by substituting $dY$ from \eqref{eq:SBR}, using the stochastic differential multiplication formula \eqref{eq:corr} and the derivative calculations \eqref{eq:G-diff} above;
this a straightforward algebraic manipulation, so we leave the details to the reader.
\end{IEEEproof}

To continue, we will require a further piece of technical machinery that will act as a Lyapunov function for the dynamics \eqref{eq:BR}\,/\,\eqref{eq:SBR}.
Specifically, following \cite{MS16,BM14}, consider the so-called \emph{Fenchel coupling}
\begin{equation}
\label{eq:Fenchel}
\fench(x,y)
	= h(x) + h^{\ast}(y) - \braket{y}{x}.
\end{equation}
By the convexity of $h$, $\fench(x,y)$ is convex in both $x$ and $y$;
furthermore, by Fenchel's inequality \cite{Roc70}, we have $\fench(x,y) \geq 0$ with equality if and only if $x = \nabla h^{\ast}(y)$, i.e. if and only if $x = \gibbs(y)$.
Moreover, using the explicit log-sum-exp formula \eqref{eq:conjugate}, it is easy to see that
\begin{equation}
\label{eq:Fench-KL}
\fench(x,y)
	= \dkl(x,\gibbs(y))
	\quad
	\text{for all $x\in\feas$, $y\in\R^{\paths}$}.
\end{equation}
As a result, for fixed $\eq\in\feas$, we will have $x(t) \to \eq$ if and only if $\fench(\eq,y(t)) \to 0$.

With this in mind, we turn to the behavior of $\fench$ under \eqref{eq:SBR}:

\begin{lemma}
\label{lem:Fenchel}
Let $Y(t)$ be given by \eqref{eq:SBR}.
Then, for all $\eq\in\feas$, we have:
\begin{equation}
\label{eq:dF}
d\fench(\eq,Y)
	= \insum_{\path} (X_{\path} - \eq_{\path}) \dd Y_{\path}
	+ \frac{1}{2}
	\insum_{\path} X_{\path} \left( \noisevar_{\path\path} - \insum_{\pathalt\sim\path} X_{\pathalt} \noisevar_{\path\pathalt} \right) \dd t
\end{equation}
\end{lemma}

\begin{IEEEproof}
Simply apply Itô's lemma \cite{Oks07,Kuo06} to $\fench(\eq,Y(t))$ and use the differentiation formula \eqref{eq:G-diff} to rewrite $\pd_{\path}\pd_{\pathalt} h^{\ast}(y) = \pd_{\pathalt} \gibbs_{\path}(y)$ in terms of $x$;
the details are straightforward and are left to the reader.
\end{IEEEproof}

\medskip

Finally, we will require the following growth estimate for martingales with bounded volatility:

\begin{lemma}
\label{lem:Wbound}
Let $W(t)$ be a Wiener process in $\R^{n}$ and let $Z(t)$ be a bounded, continuous process in $\R^{n}$.
Then, for every positive function $f$ such that $f(t)/\sqrt{t \log\log t} \to \infty$ as $t\to\infty$, we have:
\begin{equation}
\label{eq:Wbound}
f(t) + \int_{0}^{t} Z(s) \cdot dW(s)
	\sim f(t)
	\quad
	\text{as $t\to\infty$ \as}.
\end{equation}
\end{lemma}

\begin{proof}
Let $\xi(t) = \sum_{i=1}^{n} \int_{0}^{t} Z_{i}(s) \dd W_{i}(s)$.
Then, the quadratic variation $\rho = [\xi,\xi]$ of $\xi$ satisfies
\begin{equation}
\label{eq:covest1}
d[\xi,\xi]
	= d\xi \cdot d\xi
	= \insum_{i=1}^{n} Z_{i} Z_{j} \delta_{ij} \dd t
	\leq M \dd t,
\end{equation}
where $M = \sup_{t\geq0} \norm{Z(t)}^{2} < \infty$.
By the time-change theorem for martingales \citep[Cor.~8.5.4]{Oks07}, there exists a Wiener process $\wilde W(t)$ such that $\xi(t) = \wilde W(\rho(t))$, so $[f(t) + \xi(t)] / f(t) = 1 + \wilde W(\rho(t))/f(t)$.
If $\lim_{t\to\infty} \rho(t) < \infty$,
we have $\limsup_{t\to\infty} \abs{\wilde W(\rho(t))} < \infty$ (a.s.);
otherwise, if $\lim_{t\to\infty} \rho(t) = \infty$,
our claim follows from the law of the iterated logarithm \cite{Oks07} and the fact that $\rho(t) \leq Mt$.
\end{proof}

%% file: App-Deterministic.tex

\begin{IEEEproof}[Proof of Theorem \ref{thm:conv-det}]
Let $x(t)$ be an interior solution orbit of the Boltzmann routing dynamics \eqref{eq:BR} and let $\eqset = \argmin \obj$ denote the solution set of \eqref{eq:CM}.
By Lemma \ref{lem:Fenchel} and the convexity of the objective function $\obj$, it follows that the function $V(t) = \fench(\eq,y(t))$ is decreasing for any $\eq\in\eqset$, so Lyapunov's theorem ensures that $x(t) = G(y(t))$ converges to $\eqset$ (note here that all stochastic terms in \eqref{eq:dF} vanish when $\noisedev=0$).
Assume now that $x(t)$ has two distinct $\omega$-limit points $\eq, x^{\ast\ast}\in\eqset$ such that $x(t_{n})\to \eq$ and $x(t_{n}')\to x^{\ast\ast}$ for two time sequences $t_{n}$, $t_{n}' \nearrow\infty$.
It then follows that $\fench(\eq,y(t_{n}))\to 0$ and $\fench(x^{\ast\ast}, y(t_{n}'))\to0$;
however, with both $\fench(\eq,y(t))$ and $\fench(x^{\ast\ast},y(t))$ decreasing, we obtain $\lim \fench(\eq,y(t)) = \lim \fench(x^{\ast\ast},y(t)) = 0$.
This only holds if $\eq = x^{\ast\ast}$, so we conclude that $x(t)$ has a unique $\omega$-limit.
\end{IEEEproof}

%% file: App-Averages.tex

\begin{IEEEproof}[Proof of Theorem \ref{thm:averages}]
Let $\eq$ be a solution of \eqref{eq:CM} and consider the $\temp$-deflated Fenchel coupling
\begin{equation}
\label{eq:Hp}
H
	\equiv \frac{1}{\temp} \fench(\eq,\temp Y)
	= \frac{1}{\temp} \cdot\left[ h(\eq) + h^{\ast}(\temp Y) - \braket{\temp Y}{\eq} \right].
\end{equation}
Then, the Itô formula of Lemma \ref{lem:Fenchel} readily gives:
\begin{flalign}
\label{eq:dH1}
dH
	&= -\frac{\dot\temp}{\temp} H \dd t
	+ \frac{\dot\temp}{\temp} \braket{Y}{X - \eq} \dd t
	+ \braket{dY}{X - \eq}
	+ \frac{\temp}{2} \insum_{\pathalt} \frac{\pd^{2} h^{\ast}}{\pd y_{\pathalt}^{2}} \sigma_{\pathalt}^{2} \dd t.
\end{flalign}
Hence, letting $\payv = -\modcost$ and combining the definition of $H$ with the dynamics \eqref{eq:SBR}, we get:
\begin{flalign}
\label{eq:dH2}
dH
	= - \frac{\dot\temp}{\temp^{2}} \left[ h(\eq) - h(X) \right] dt
	+ \braket{\payv}{X - \eq} dt
	+ \insum_{\pathalt} (X_{\pathalt} - \eq_{\pathalt})\,\sigma_{\pathalt} \dd W_{\pathalt}
	+ \frac{\temp}{2} \insum_{\pathalt} \frac{\pd^{2}h^{\ast}}{\pd y_{\pathalt}^{2}} \sigma_{\pathalt}^{2} \dd t,
\end{flalign}
where we used the fact that $h^{\ast}(\temp Y) = \braket{\temp Y}{X} - h(X)$.
Thus, after rearranging and integrating:
\begin{subequations}
\label{eq:reg}
\begin{flalign}
\int_{0}^{t} \braket{\payv(s)}{\eq - X(s)} \dd s
	&\label{eq:cost-Hdiff}
	= H(0) - H(t)
	\\
	&\label{eq:cost-temp}
	-\int_{0}^{t} \frac{\dot\temp(s)}{\temp^{2}(s)} \left[ h(\eq) - h(X(s)) \right] ds
	\\
	&\label{eq:cost-noise}
	+ \insum_{\pathalt} \int_{0}^{t} (X_{\pathalt}(s) - \eq) \,\sigma_{\pathalt}(s) \dd W_{\pathalt}(s)
	\\
	&\label{eq:cost-Ito}
	+ \frac{1}{2} \insum_{\pathalt} \int_{0}^{t} \temp(s) \frac{\pd^{2} h^{\ast}}{\pd y_{\pathalt}^{2}} \sigma_{\pathalt}^{2}(s) \dd s.
\end{flalign}
\end{subequations}

We now proceed to bound each term of \eqref{eq:reg}:
\begin{enumerate}
[\itshape a\upshape)]

\item
Since $H\geq0$, the term \eqref{eq:cost-Hdiff} is bounded from above as follows:
\begin{equation}
\label{eq:cost-Hdiff1}
H(0)
	\leq \frac{h(\eq) + h^{\ast}(Y(0))}{\temp(0)}
	= \frac{h(\eq) - \min_{x\in\strat} h(x)}{\temp(0)} + \bigoh(1)
	= \frac{\sum_{\source} \log\nPaths_{\source}}{\temp(0)} + \bigoh(1).
\end{equation}
\item
For \eqref{eq:cost-temp},
we have $h(\eq) - h(X(s)) \leq \sum_{\source} \log\nPaths_{\source}$ by definition;
hence, with $\dot\temp \leq 0$, we get:
\begin{equation}
\eqref{eq:cost-temp}
	\leq - \insum_{\source} \log\nPaths_{\source} \int_{0}^{t} \frac{\dot\temp(s)}{\temp^{2}(s)} \dd s
	= \insum_{\source} \log\nPaths_{\source} \left( \frac{1}{\temp(t)} - \frac{1}{\temp(0)} \right).
\end{equation}

\item
By Lemma \ref{lem:Wbound} and the law of the iterated logarithm, \eqref{eq:cost-noise} is bounded from above by $2\noisebound^{2} \sqrt{t \log \log t}$.

\item
Finally, using the derivative calculations \eqref{eq:G-diff} for $\pd_{\pathalt}^{2} h^{\ast}(y) = \pd_{\pathalt} \gibbs_{\pathalt}(y)$, we immediately deduce that \eqref{eq:cost-Ito} is bounded from above by $\frac{1}{2} \noisebound^{2} \int_{0}^{t} \temp(s) \dd s$.
\end{enumerate}
The bound \eqref{eq:cost-bound} then follows by recalling that $\payv = -\modcost = -\pd\cost$, integrating the convexity relation $\obj(X(t)) \leq \obj^{\ast} + \braket{\payv(X(t))}{\eq - X(t)}$, and combining the above terms.
\end{IEEEproof}

%% file: App-Strict.tex

\renewcommand{\dual}{\R^{\paths}}
\renewcommand{\dnorm}{\norm}

To prove Theorem \ref{thm:strict} we will need a series of auxiliary results.
We begin by showing that neighborhoods of a strict, non-mixing optimum state are recurrent under \eqref{eq:SBR}:

\begin{proposition}
\label{prop:recurrence}
Let $\eq$ be a strict, non-mixing optimum state of \eqref{eq:CM}.
If $\temp$ is chosen sufficiently small, there exists a \textup(random\textup) sequence of times $t_{n}\nearrow\infty$ such that $\norm{X(t_{n}) - \eq} < \delta$ \as.
\end{proposition}

\begin{IEEEproof}
Suppose there exists some $t_{0}$ such that $\norm{X(t) - \eq} \geq \delta$ for all $t\geq t_{0}$.
Then, \eqref{eq:dH2} yields
\begin{flalign}
H(t)
	&= H(t_{0})
	+ \int_{t_{0}}^{t} \braket{\payv(X(s))}{X(s) - \eq} \dd s
	+ \frac{\temp}{2} \int_{t_{0}}^{t} \frac{\pd^{2} h^{\ast}}{\pd y_{\pathalt}^{2}} \sigma_{\pathalt}^{2} \dd s
	+ \xi(t)
	\notag\\
	&\leq H(t_{0})
	- \mu \delta(t - t_{0})
	+ \frac{\temp\noisevar}{2K} (t - t_{0})
	+ \xi(t)
	\leq H(t_{0})
	- \parens*{\mu\delta - \frac{\temp\noisevar}{2K} - \frac{\xi(t)}{t - t_{0}}} (t - t_{0}),
\end{flalign}
where, $\payv = -\modcost$ as before, and we set $\xi(t) = \sum_{i=1}^{d} \int_{t_{0}}^{t} (X_{i}(s) - \eq_{i}) \dd Z_{i}(s)$, and $\mu = \min_{\pathalt\sim\supp(\eq)} \{\payv(\eq) - \payv_{\pathalt}(\eq)\} > 0$ (recall that $\eq$ is a strict equilibrium).
By Lemma \ref{lem:Wbound}, it follows that $\xi(t)/t \to 0$ \as, so the above estimate yields $\lim_{t\to\infty} H(t) = -\infty$ if $\temp\noisevar < 2\mu\delta K$, a contradiction (recall that $H(t)\geq0$ for all $t\geq0$).
This shows that $t_{0}=\infty$ \as and our claim follows.
\end{IEEEproof}

\smallskip

We now show that if $X(t)$ begins close enough to a strict equilibrium $\eq$, then it remains nearby and eventually converges with arbitrarily high probability:

\begin{proposition}
\label{prop:asymstab}
Let $\eq$ be a strict, non-mixing optimum state of \eqref{eq:CM}.
Then, for every $\eps>0$ and for every neighborhood $U_{0}$ of $\eq$ in $\strat$, there exists a neighborhood $U\subseteq U_{0}$ of $\eq$ such that, if $X(0)\in U$, then:
\begin{equation}
\label{eq:stable-asym}
\txs
\prob\left(
	\text{$X(t)\in U_{0}$ for all $t\geq0$ and $\lim_{t\to\infty} X(t) = \eq$}
	\right)
	\geq 1-\eps.
\end{equation}
\end{proposition}

\begin{IEEEproof}
Write $\eq = (\path_{1}^{\ast},\dotsc,\path_{\nSources}^{\ast})$ and let $\act_{\source}^{\ast} \equiv \act_{\source}\setminus\{\path_{\source}^{\ast}\}$ denote the set of non-equilibrium paths of the $\source$-th source.
Moreover, for all $\path\in\paths_{\source}^{\ast}$, let
\begin{equation}
\label{eq:Zdef}
Z_{\path}
	= \temp \left( Y_{\path} - Y_{\path_{\source}^{\ast}} \right),
\end{equation}
so that, by definition, $X(t)\to\eq$ if and only if $Z_{\path}(t) \to -\infty$ for all $\path\in\act_{\source}^{\ast}$.
With this in mind, let $M>0$ be sufficiently large so that $X(t) \in U$ if $Z_{\path}(t) \leq -M$ for all $\path\in\act_{\source}^{\ast}$, and assume that
\begin{equation}
M
	> \mu \temp \noisebound^{2} \log(N/\eps),
\end{equation}
with $\mu$ defined as in the proof of Proposition \ref{prop:recurrence}.
We will show that if $Z_{\path}(0) \leq - 2M$, then $X(t) \in U$ for all $t\geq0$ and $Z_{\path}(t)\to-\infty$ with probability at least $1-\eps$.

To that end, let $\tau_{U} = \inf\{t>0: X(t) \notin U\}$ denote the first exit time from $U$;
then, for $t\leq\tau_{U}$, \eqref{eq:SBR} gives
\begin{flalign}
\label{eq:Z1}
Z_{\path}(t)
	= Z_{\path}(0)
	+ \temp \int_{0}^{t} \left[ \payv_{\path}(X(s)) - \payv_{\path_{\source}^{\ast}}(X(s)) \right] \dd s
	+ \temp \xi_{\source}(t)
	\leq -2M - \temp \left[ \mu t - \xi_{\source}(t) \right],
\end{flalign}
where we have set
\begin{equation}
\label{eq:diffnoise-strict}
\xi_{\source}(t)
	= \int_{0}^{t} \sigma_{\path}(X) \dd W_{\path} - \int_{0}^{t} \sigma_{\path_{\source}^{\ast}}(X) \dd W_{\path_{\source}^{\ast}}.
\end{equation}
We will first show that $\prob(\tau_{U} < \infty) \leq \eps$.
Indeed, as in the proof of Lemma \ref{lem:Wbound},
there exists a rescaled Wiener process $\wilde W_{\source}(t)$ such that $\xi_{\source}(t) = \wilde W_{\source}(\rho_{\source}(t))$ where $\rho_{\source} = [\xi_{\source},\xi_{\source}]$ is the quadratic variation of $\xi_{\source}$.
We thus conclude that $Z_{\path}(t) \leq -M$ whenever $\mu t - \wilde W_{\source}(\rho_{\source}(t)) \geq - M/\temp$.
Accordingly, with $\rho_{\source}(t) \leq 2\noisebound^2t$,
it suffices to show that the hitting time
\begin{equation}
\label{eq:hitline}
\tau_{0}
	= \inf\left\{
	t>0: \wilde W_{\source}(t) = \frac{\mu t}{2\noisebound^2} + \frac{M}{\temp}
	\;\text{for some $k\in\play$}
	\right\}
\end{equation}
is finite with probability not exceeding $\eps$.

Now, if a trajectory of $\wilde W_{\source}$ has $\wilde W_{\source}(t) \leq \mu t/(2\noisebound^2) + M/\temp$ for all $t\geq0$, we also get
\begin{equation}
\wilde W_{\source}(\rho_{\source}(t))
	\leq \frac{\mu  \rho_{\source}(t)}{2 \noisebound^2} + \frac{M}{\temp}
	\leq \mu t + \frac{M}{\temp},
\end{equation}
so $\tau_{U}$ is infinite whenever $\tau_{0}$ is infinite.
Thus, if we write $E_{\source}$ for the event that $\wilde W_{\source}(t) \geq \mu t/(2\noisebound^2) + M/\temp$ for some finite $t\geq0$, the hitting-time estimate \cite[p.~197]{KS98} for a Brownian motion with drift yields $\prob(E_{\source}) = e^{-\lambda_{\source} M}$ with $\lambda_{\source} = \mu  / (\temp \noisebound^2)$.
Therefore, by construction of $M$, we obtain:
\begin{equation}
\prob(\tau_{0} < +\infty)
	= \prob\left( \union\nolimits_{\source} E_{\source} \right)
	\leq \insum_{\source}  \prob(E_{\source})
	= \insum_{\source}  e^{-\lambda_{\source}M}
	\leq \eps.
\end{equation}
By Lemma \ref{lem:Wbound}, Eq.~\eqref{eq:Z1} then yields
\(
Z_{\path}(t)
	\leq -2M - \temp \big[ \mu t - \xi_{\source}(t) \big]
	\sim - \temp \mu t
	\to-\infty
	\quad
	\text{(a.s.)}.
\)
We conclude that $Z_{\path}(t)\to-\infty$ for all $\path\in\act_{\source}^{\ast}$, so $\lim_{t\to\infty} X(t) = \eq$, as claimed.
\end{IEEEproof}

With all this at hand, we are finally in a position to prove Theorem \ref{thm:strict}:

\begin{IEEEproof}[Proof of Theorem \ref{thm:strict}]
By Proposition \ref{prop:recurrence}, if $\temp$ is chosen sufficiently small, $X(t)$ will visit any neighborhood $U$ of $\eq$ infinitely often.
Thus, if $\eps>0$ is fixed and $U$ is chosen as in the statement of Proposition \ref{prop:asymstab}, $X(t)$ will stay in $U$ and converge to $\eq$ with probality exceeding $1-\eps$.
Since $X(t)$ is recurrent, the probability that $X(t)$ exits and then re-enters $U$ for a strictly positive amount of time at least $n$ times does not exceed $\eps^{n}$.
Our claim then follows by taking $n\to\infty$.
\end{IEEEproof}

%% file: App-Interior.tex

\begin{IEEEproof}[Proof of Theorem \ref{thm:interior}]
Assume that $\eq$ is an isolated interior solution of \eqref{eq:CM} and define $H(t) = \temp^{-1} \fench(\eq,\temp Y(t))$ as in the proof of Theorem \ref{thm:strict}.
Then, taking a sufficiently small $B>0$ such that $\obj(x) - \obj(\eq) \geq \frac{1}{2} B \norm{x - \eq}^{2}$ (recall that $\eq$ is an isolated, interior minimum point), \eqref{eq:reg} yields
\begin{equation}
H(t)
	\leq H(0)
	- \frac{1}{2} B \int_{0}^{t} \norm{X(t) - \eq}^{2} \dd s
	+ \frac{\temp}{2} \noisebound^{2} t
	+ \xi(t),
\end{equation}
where the martingale term $\xi(t)$ is defined as in the proof of Proposition \ref{prop:recurrence}.
Thus, dividing by $t$ and taking expectations yields the bound
\begin{equation}
\label{eq:average}
\exof*{\frac{1}{t} \int_{0}^{t} \norm{X(s) - \eq} \dd s}
	\leq \frac{\temp \noisebound^{2}}{B} + \frac{H(0)}{t}.
\end{equation}
Our goal will be to show that $X(t)$ admits an invariant measure which can be estimated from the above bound to yield \eqref{eq:invariant}.

To that end, consider the process $\Psi(X(t))$ where $\Psi_{\source\pathalt}(x) = \log (x_{\source\pathalt}/x_{\source\path}) = Y_{\source\pathalt} - Y_{\source\path}$ for some fixed $\path\in\paths_{\source}$, $\pathalt\sim\path$, $\source\in\sources$.
Itô's formula then yields:
\begin{equation}
\label{eq:dY}
d\Psi_{\pathalt}
	=dY_{\pathalt} - dY_{\path}
	=\gen\Psi_{\pathalt}(X)\dd t
	+ \insum_{\edge} \pmat_{\edge\pathalt}' \noisedev_{\edge} \dd W_{\edge},
\end{equation}
where $\gen\Psi$ denotes the deterministic part of \eqref{eq:Z1}
and
$\pmat_{\edge\pathalt}' = \pmat_{\edge\pathalt} - \pmat_{\edge\path}$ (the source index $\source$ has been dropped for convenience).
Following \cite{Imh05}, we claim that the infinitesimal generator of $\Psi$ is (uniformly) elliptic \cite{Oks07,Kuo06}.
For this,
let $A_{\pathalt \edge} = \pmat_{\edge\pathalt}' \noisedev_{\edge}$, $\pathalt\in\coprod_{\source}\act_{\source}^{*}$ denote the coefficient matrix of the martingale term of \eqref{eq:dY};
then, it suffices to show that the matrix $AA^{T}$ is positive-definite on $\strat$.
Indeed, for all $z\in \tcone_{\strat}(\eq)$, we have:
\begin{equation}
\label{eq:eigenvalue}
\langle Az, Az\rangle
	=\insum_{\pathalt,\gamma}\insum_{\edge} \pmat_{\edge\pathalt}' \pmat_{\edge\gamma}' \noisedev_{\edge}^{2} z_{\pathalt} z_{\gamma}
	=\insum_{\edge} \noisedev_{\edge}^{2} (\pmat' z)_{\edge}^{2}
	> 0,
\end{equation}
because $\pmat'$ is necessarily invertible on $\feas$ (otherwise $\eq$ would not be isolated).
This proves our ellipticity assertion,
so \cite[Lemma 3.4]{Bha78} shows that $\Psi$ is recurrent.

Since $\Psi$ is recurrent, it follows that $X(t)$ is also recurrent;
consequently, the transition probabilities of $X(t)$ converge in total variation to an invariant probability measure $\nu$ on $\strat$ \cite{Kuo06,Kha12}.
Thus, if $B_{\delta}$ is a $\delta$-ball centered at $\eq$, ergodicity gives $\nu(B_{\delta}) = \lim_{t\to\infty}\ex_{x}\left[\frac{1}{t}\int_{0}^{t}\chi_{B_{\delta}}(X(s))\dd s\right]$, where $\chi_{B_{\delta}}$ is the indicator function of $B_{\delta}$.

Now, with $\norm{x - \eq} \geq \delta$ outside $B_{\delta}$, we readily obtain
\begin{equation}
\exof*{ \frac{1}{t} \int_{0}^{t} \chi_{B_{\delta}}(X(s))\dd s}
	\geq \exof*{\frac{1}{t}\int_{0}^{t} \left(1 -\norm{X(s) - \eq}^{2} / \delta^{2}\right) \dd s}.
\end{equation}
Thus, letting $t\to\infty$ in \eqref{eq:average}, we get $\nu(B_{\delta}) \geq 1 - \temp\noisebound^{2} / (B\delta^{2})$.
Taking $\temp$ sufficiently small, this bound can be assumed greater than or equal to $1-\eps$;
our assertion then follows from the convergence in total variation of the transition probabilities of $X(t)$ to the invariant measure $\nu$.
\end{IEEEproof}